%% file: main.tex
\documentclass[sigconf]{acmart}
\setcopyright{none}
\renewcommand\footnotetextcopyrightpermission[1]{}

\usepackage{graphicx}
\usepackage{amsmath,amsfonts}
\usepackage{booktabs}
\usepackage{caption}
\usepackage{enumitem}
\setlist[description]{style=unboxed}
\usepackage{hyperref}
\usepackage{multirow}
\usepackage{adjustbox}
\usepackage{tikz}
\usetikzlibrary{fit,cd, backgrounds}
\usepackage{tabularx}
\usepackage{algorithm}
\usepackage{algpseudocode}
\usepackage{pgfplots}
\usepgfplotslibrary{fillbetween}
\usepackage{bbding}
\usepackage{subcaption}
\usepackage{balance}

\newcounter{protocol}
\makeatletter
\newenvironment{protocol}[1][htb]{%
  \let\c@algorithm\c@protocol
  \renewcommand{\ALG@name}{Protocol}
  \begin{algorithm}[#1]%
  \begin{flushleft}
  }{\end{flushleft}
  \end{algorithm}
}
\makeatother

\newcommand{\change}[1]{#1}

\usepackage{amsthm,thm-restate}
\newtheorem{claim}{Claim}[section]

\newcommand{\keyset}{\mathcal{K}}
\newcommand{\valset}{\mathcal{V}}
\newcommand{\compnodes}{\ell}
\newcommand{\compsub}{t}
\newcommand{\geonoise}{x_k}
\newcommand{\geoparam}{r}
\newcommand{\freqk}{q_k}
\newcommand{\meank}{\mu_k}
\newcommand{\sumk}{S_k}
\newcommand{\freqnoise}{\epsilon_F}
\newcommand{\meannoise}{\epsilon_M}
\newcommand{\leaknoise}{\epsilon_L}
\newcommand{\numkeys}{n}
\newcommand{\subsetratio}{p}
\newcommand{\numkv}{|\mathcal{S}|}
\newcommand{\freq}{q}
\newcommand{\mean}{\mu}
\newcommand{\freqalg}{F}
\newcommand{\meanalg}{M}
\newcommand{\leakage}{L}
\newcommand{\observation}{Z}
\newcommand{\distinctkeys}{\lambda}
\newcommand{\minfreq}{\gamma}
\newcommand{\allkv}{\mathcal{S}}
\usepackage{cleveref}

\begin{document}

\title[Selective MPC]{Selective MPC: Distributed Computation of Differentially Private Key-Value Statistics}
\date{}
\author{Thomas Humphries}
\authornotemark[1]
\affiliation{University of Waterloo}
\email{thomas.humphries@uwaterloo.ca}

\author{Rasoul Akhavan Mahdavi}
\authornotemark[1]
\affiliation{University of Waterloo}
\email{rasoul.akhavan.mahdavi@uwaterloo.ca}

\author{Shannon Veitch}
\authornote{These authors contributed equally to this research.}
\affiliation{University of Waterloo}
\email{ssveitch@uwaterloo.ca}

\author{Florian Kerschbaum}
\affiliation{University of Waterloo}
\email{florian.kerschbaum@uwaterloo.ca}

\renewcommand{\shortauthors}{T. Humphries et al.}


\begin{abstract}
\input{abstract}
\end{abstract}



\maketitle


\section{Introduction}
\input{Intro.tex}

\section{Preliminaries}\label{sec.prelims}
\input{Prelim.tex}

\section{Problem Statement}\label{sec.prob}
\input{Problem_state.tex}

\section{Preliminary Solutions}\label{sec.early}
\input{Initial_soln.tex}

\section{Protocol Description}\label{sec.prot}
\input{Protocol_desc.tex}

\section{Privacy}\label{sec.priv}
\input{Privacy.tex}

\section{Security}\label{sec.sec}
\input{Security.tex}

\section{Evaluation}\label{sec.eval}
\input{Evaluation.tex}

\section{Discussion}\label{sec.future}
\input{Discussion.tex}

\section{Related Work}\label{sec.rel}

\input{Related.tex}

\section{Conclusion}
\input{Conclusion.tex}

\begin{acks}
We gratefully acknowledge the support of NSERC for grants RGPIN-05849, IRC-537591, the NSERC Postgraduate Scholarship-Doctoral program, and the Royal Bank of Canada for funding this research.
\end{acks}

\bibliographystyle{ACM-Reference-Format}
\balance
\bibliography{references}
\appendix
\input{appendix.tex}

\end{document}

%% file: abstract.tex
Key-value data is a naturally occurring data type that has not been thoroughly investigated in the local trust model.
Existing local differentially private (LDP) solutions for computing statistics over key-value data suffer from the inherent accuracy limitations of each user adding their own noise.
Multi-party computation (MPC) maintains better accuracy than LDP and similarly does not require a trusted central party. However, naively applying MPC to key-value data results in prohibitively expensive computation costs.
In this work, we present \emph{selective multi-party computation}, a novel approach to distributed computation that leverages DP leakage to efficiently and accurately compute statistics over key-value data.
By providing each party with a view of a random subset of the data, we can capture subtractive noise. We prove that our protocol satisfies pure DP and is provably secure in the combined DP/MPC model. 
Our empirical evaluation demonstrates that we can compute statistics over 10,000 keys in 20 seconds and can scale up to 30 servers while obtaining results for a single key in under a second.

%% file: Intro.tex
Key-value data naturally occurs in many internet applications, such as video ad and mobile app analytics where keys are identifiers and values are time or frequencies~\cite{ye19}.
In these applications, a provider collects key-value data from a large number of users to compute aggregate statistics which provide crucial insights.
However, as individuals become more privacy-aware, they become increasingly unwilling to divulge their sensitive data to a central party. Consequently, recent work has investigated systems that simultaneously allow for the computation of useful statistics while respecting the privacy of individuals and their data. A popular method for achieving this goal is to use local differential privacy (LDP). Under this model, users take privacy into their own hands, so to speak, by perturbing their data before sending it to a central party to be processed. This technique provides a limit on the amount of information the central party can infer about each user's data while allowing the central party to compute a meaningful aggregate result. These properties have made local differential privacy a preferential solution for many organizations such as Google~\cite{google}, Apple~\cite{apple}, Mozilla~\cite{mozilla}, and Microsoft~\cite{msr17} to collect user statistics.
For example, the Chrome Web browser uses local differential privacy to collect aggregate statistics about settings such as the homepage and default search engine to counter software that maliciously alters these settings~\cite{google}. As another example, Mozilla experimented with using local differential privacy techniques to collect browser telemetry data in a private manner~\cite{mozilla}.

Applying LDP to key-value data is more challenging than numerical or categorical data.
Naively applying noise to the key implies a high loss in the accuracy of aggregate statistics, since it breaks the correlation between the key and the value. Both Ye et al.~\cite{ye19} and later Gu et al.~\cite{gu2019pckv} develop improved LDP mechanisms that reduce this accuracy loss.
While improving significantly over naive solutions, this line of work is still confined by the inherent accuracy limitations of the LDP model. Intuitively, this is because each user individually adds randomness, which compounds in the final aggregated result. 

An approach for achieving both the accuracy of the central model in DP and the trust of the local model is to use multi-party computation (MPC). 
A popular type of MPC uses secret sharing, wherein the parties jointly compute a function without any one party learning the input.
Applying MPC to key-value data is challenging because it may introduce a prohibitively large computational overhead to keep the inputs concealed.
Specifically, simultaneously computing over the keys while keeping them secret is a complex issue.
To compute and update aggregate statistics, we require an array which maintains the totals for each key.
If the key is secret shared, and thus kept confidential from each party, then indexing each value into this array using MPC incurs a cost linear in the size of the array.
This results in a total computation cost on the order of $\numkv \numkeys$, where $\numkv$ and $\numkeys$ denote the total number of key-value pairs and the number of distinct keys, respectively.
An alternative approach, using ORAM \cite{mazloom18}, has polylogarithmic cost; however, requires large constants and is typically applied to huge data set sizes.
Hence, these approaches remain impractical.
If the key is reveled to the parties or encrypted in a way that reveals the access pattern (e.g., public-key searchable symmetric encryption \cite{curtmola2011searchable}), the parties will learn the count of each keyword, which violates DP.
To prevent this violation, one could pad the counts by adding dummy key-value pairs, restoring an approximate differentially private guarantee. This requires adding an abundance of dummy records in order to achieve any reasonable privacy guarantees. Thus, this approach incurs a large communication cost.

In this paper, we design a new method of MPC called \emph{selective multi-party computation}. Selective MPC does not conceal keys, resulting in an efficient protocol which does not incur prohibitively high communication costs.
In typical MPC solutions, all parties receive shares for all key-value pairs. In contrast, our design involves randomly selecting a subset of parties to receive shares for each key-value pair. Thus, each party only views shares for a random subset of the entire dataset. In other words, our approach simulates removing some data from the view of each party, capturing subtractive noise (similar to how adding dummies in previous approaches captured additive noise). 
We compose the two approaches (selecting a subset of parties and adding dummies) to emulate a two-sided distribution without removing any data. We obtain a new theoretical result, proving that one can achieve pure differential privacy from joining two different one-sided distributions.

Our approach substantially reduces the computation costs, since our novel construction enables the keys and access patterns to be revealed while preserving DP guarantees.
Compared to the naive MPC approach, the computation complexity of the MPC is reduced by a factor of $O(\numkv)$ where $\numkv$ is the number of key-value pairs.
By simulating both subtractive and additive noise, selective MPC requires a very low number of dummy records.
Compared to the padding approach, we produce 10-20 times less dummies and achieve a low $\epsilon$ in a pure-DP model ($\delta = 0$).
Compared to the LDP mechanisms by Ye et al.~and Gu et al.~we improve accuracy by a factor of $O(\sqrt{m})$ where $m$ is the number of clients.
We give a detailed comparison to other approaches in \Cref{sec.eval}.

In summary, our contribution is a secure multi-party computation of differentially private frequency and mean estimations over key-value data with the following properties.
\begin{itemize}\itemsep0em
    \item Distributed trust and provable security guarantees.
    \item Pure differential privacy in the local trust model.
    \item High accuracy independent of the number of users.
    \item Scalability up to 10,000 keys with results in at most 20 seconds. 
\end{itemize}

%% file: Prelim.tex
\subsection{Differential Privacy}\label{sec.dp}
Differential privacy (DP), introduced by Dwork et al.~\cite{dwork-dp}, is a common privacy notion that enables the calculation of aggregate statistics on users' data in a privacy-preserving manner. Intuitively, differential privacy states that a user's participation has a bounded effect on the final output. More formally, differential privacy can be defined as follows.

\begin{definition}[Differential Privacy]\label{def.adp}
A randomized algorithm $M:\mathcal{X}^n\mapsto \mathcal{Y}$ is $(\epsilon,\delta)$-DP, if for any pair of neighbouring datasets $X,X' \in \mathcal{X}^n$, and for any $T \subseteq \mathcal{Y}$ we have
\begin{equation}
    \Pr[M(X)\in T] \leq e^{\epsilon} \Pr[M(X')\in T] + \delta.
\end{equation}
\end{definition}

If $\delta \neq 0$, then we say that the mechanism provides approximate differential privacy. When $\delta = 0$ it satisfies pure differential privacy.
There are two main trust models in differential privacy. In the central model, we assume the existence of a trusted curator that gathers the users' data, calculates the noisy aggregate statistic, and publishes it. In the local model, users send noisy data to a (potentially untrusted) curator to calculate the aggregate statistic. Although this model eliminates the need to trust the aggregator, it comes at the cost of utility. This is because there is noise added to each data point on the order of the number of clients.

The most common approach to satisfying DP is to add random noise to the output of a function. To prove that this satisfies DP, the function must have a bounded output. We formalize this notion as the sensitivity of a function.

\begin{definition}[Sensitivity]
    Let $f:\mathcal{X}^n\mapsto \mathbb{R}^k$. If $D$ is a distance metric between elements of $\ \mathbb{R}^k$ then the $D$-sensitivity of $f$ is
    \begin{equation}
        \Delta^{(f)}_D=\max_{(X,X')} D(f(X), f(X')),
    \end{equation}
    where $(X,X')$ are pairs of neighbouring datasets.
\end{definition}
Common examples of the distance metric are the $\ell_1$-norm and $\ell_2$-norm. With the notion of sensitivity, we can now define the Laplace mechanism, which uses sensitivity based on the $\ell_1$-norm, which we denote $\Delta^{(f)}_1$ for simplicity.

\begin{definition}[Laplace Mechanism]
    Let $f:\mathcal{X}^n\mapsto \mathbb{R}^k$. The Laplace mechanism is defined as 
    \begin{equation}
        M(X) = f(X) + Lap(\Delta^{(f)}_1/\epsilon)^\numkeys,
    \end{equation}
    where $Lap(\Delta^{(f)}_1/\epsilon)^\numkeys$ represents $\numkeys$ i.i.d draws from the Laplace distribution with parameter $\Delta^{(f)}_1/\epsilon$.
\end{definition}

It is well known that the Laplace mechanism satisfies $\epsilon$-DP~\cite{dwork14}. 

\subsection{Secret Sharing}
A secret sharing scheme \cite{shamir79,blakley79} enables a secret $s$ to be split into $n$ shares that may be distributed among a set of $n$ participants. The secret is split in such a way that the $n$ shares permit reconstruction of $s$, but no set of $n-1$ or fewer shares yield any information about the secret. 

For example, consider the additive secret sharing scheme. In this scheme, to split the secret $s$, we select $n-1$ shares $s_1,\dots,s_{n-1}$ at random from the domain and let $s_n = s - \sum_{i=1}^{n-1}s_i$. Then all players can recover the secret $s$ by computing $s = \sum_{i=1}^{n}s_i$. Note that each share is computed as a fixed linear function of the secret and random elements of a field, so participants can locally compute any linear function of shared values. We use $[x]$ to denote that the value of $x$ is shared among participants. For example, given constants $v_1,v_2,v_3$ and shares of values $[x],[y]$, one can locally compute 
\[
v_1[x]+v_2[y]+v_3 = [v_1\cdot x + v_2\cdot y + v_3]
\]
to obtain shares of the value $v_1\cdot x + v_2\cdot y + v_3$. This property of the additive secret sharing scheme enables the construction of more general multi-party computation protocols \cite{cramer}.

\subsection{Secure Multi-party Computation}
Secure multi-party computation (MPC) allows a set of participants $\{p_1,\dots,p_n\}$, where $p_i$ holds private data $d_i$, to jointly compute a function $f(d_1,\dots,d_n)$, while keeping their input private and revealing only the output. We focus on a variant of MPC which performs operations over shares of the data \cite{bgw88}.

In MPC based on secret sharing, participants distribute shares of their data among a set of computation nodes. The computation nodes perform a computation over the secret shared data to obtain shares of a function of the participants' data, $f(d_1,\dots,d_n)$. Most MPC protocols represent the desired function as a circuit using basic operations over a finite field (addition and multiplication) and provide subprotocols for performing those basic operations \cite{MAURER2006370,spdz}.

We can consider two different models for secure multi-party computation: the \emph{semi-honest} (or passive) model and the \emph{malicious} (or active) model. In the semi-honest model, introduced by Goldreich \cite[\S7.2.2]{goldreich09}, adversaries do not deviate from the protocol but may gather information to infer private information. In the malicious model \cite[\S7.2.3]{goldreich09}, adversaries may deviate from the protocol by altering messages or other means to infer private information.

%% file: Problem_state.tex
Our goal is to present a protocol that permits the efficient computation of differentially private statistics over key-value data while sacrificing minimal trust to any party. 
\change{A naive solution might first divide the key-value pairs into a set of keys and a set of values and then apply existing methods for computing differentially private histograms to each set. However, as Ye et al.~\cite{ye19} argue, this does not maintain the correlation between keys and values. For this reason, computing meaningful differentially private statistics over key-value data is more challenging than histograms.}

\subsection{Notation}
We have a set of $m$ clients, each denoted by an index $i \in [m]$. Each client has a set of key-value pairs $\allkv_i$. We denote the $j$th key-value pair of client $i$ to be $\langle k_{ij}, v_{ij} \rangle$, where each $k_{ij} \in \keyset$ and each $v_{ij} \in \valset$ for some sets $\keyset, \valset$. The multiset of key-value pairs owned by all $m$ clients is denoted $\allkv$. For a value or algorithm, we adopt the convention of using a subscript $k$ to restrict it to a key $k$. We summarize the notation in \Cref{tab.notation}.

\begin{table}[t]
    \centering
    \begin{tabular}{|c|c|}
         \hline
         \textbf{Symbol} & \textbf{Description}\\
         \hline
         $\mathcal{K}$ & domain of the keys \\
         $\mathcal{V}$ & domain of the values \\
         $n$ & the size of the key space $\mathcal{K}$ \\
         $m$ & the number of clients \\
         $\allkv$ & multiset of all key-value pairs \\
         $\allkv_i$ & set of key-value pairs owned by client $i$ \\
         $R$ & bound on the values in $\mathcal{V}$ \\
         $\compnodes$ & the number of computation nodes \\
         $\compsub$ & size of subset of computation nodes \\
         $\lambda$ & maximum \# of keys held by a client \\
         $\gamma$ & minimum frequency of any key \\
         $\freq$ & vector of frequencies \\
         $\mean$ & vector of means \\
         $r$ & parameter for geometric distribution \\
         $p$ & ratio of $t$/$\ell$ \\
         \hline
    \end{tabular}
    \caption{Notation}
    \label{tab.notation}
\end{table}

In this work, we focus on the following two statistics over key-value data. Let $\mathcal{A}_k = \{\langle k_{ij}, v_{ij}\rangle \in \allkv \mid k_{ij} = k\}$ be the multiset of key-value pairs corresponding to key $k$.
\begin{itemize}[itemsep=0em]
    \item \textbf{Frequency estimation.} The frequency of some key $k$, $\freqk$, is defined as the number of instances of the key $k$ in the set $\allkv$. That is, $\freqk(\allkv) = |\mathcal{A}_k|$.
    We denote the vector of frequencies of all keys by $\freq(\allkv)$.
    \item \textbf{Mean estimation.} The mean of the key $k$, $\meank$, is the mean of all values in $\allkv$ whose key is $k$. Thus, the mean is
    \[
    \meank(\allkv) = \frac{\sum_{\mathcal{A}_k} v_{ij}}{\freqk(\allkv)}.
    \]
    We denote the vector of means for all keys by $\mean(\allkv)$.
\end{itemize}

\subsection{Assumptions}
We aim to compute differentially private statistics over the data. The output will satisfy the requirements of differential privacy, as described in \Cref{sec.dp}.
In the case of telemetry data, frequency and mean very helpful statistics. For example, a browser vendor may want to calculate how many users have added known trackers to a blocklist~\cite{mozilla}.
We make the following assumptions about the data in the set $\allkv$.

\begin{itemize}
    \item \textbf{Bounded values.} We assume that the domain of the values is bounded by some interval of size $2R$. That is, every value $v_{ij} \in [x-R,x+R]$, for some $x$. We assume this to bound the sensitivity of statistics on the values.
    \item \textbf{Distinct keys.} We assume that each key appears in $\allkv_i$ at most once. This is done for simplicity because in the case where a client has multiple values for a single key, they could simply aggregate their data for said key beforehand.
    \item \textbf{Maximum size of $\allkv_i$.} We assume that each user has at most $\lambda$ key-value pairs. That is, $|\allkv_i| \leq \lambda$. This is assumed to tighten the analysis but is not restrictive because in the worst case $\lambda= \numkeys$.
    \item \textbf{Minimum frequency.} We assume that the frequency of each key is at least $\gamma$, for some $\gamma \geq 1$. This is assumed to bound the sensitivity of our mean calculation. We require this term to obtain a data-independent bound on the sensitivity since both the frequency and mean are sensitive in our setting.
\end{itemize}

We consider user-level differential privacy where neighbouring datasets change by adding or removing a single user's key-value pairs. \change{Here, the users correspond to the clients which we discuss throughout the paper.} We apply concepts from secure multi-party computation to remove the need for a single trusted server. In particular, the clients distribute their data between a small set of untrusted servers, which collectively compute statistics over the data. We call these servers the computation nodes, and we require that there are $\ell \geq 3$ computation nodes. We make the following assumptions about the entities involved in the computation.

\begin{itemize}
    \item \change{\textbf{Anonymous channels.} We assume that there exist anonymous channels between clients and computation nodes. \Cref{sec.dumgen} discusses how our protocol would function in the shuffle model, where there exists an intermediary third party between the clients and the computation nodes \cite{frist_shuffle,google_shuffle}. Our assumption of anonymous channel is a relaxation of the shuffle model assumption.}

    \item \textbf{Trust model.}
    For simplicity, we begin by describing our protocol in the semi-honest model. \change{In particular, we assume that the clients and the computation nodes follow the protocol. Any party may however attempt to infer other information from the data that they observe (semi-honest model).} \Cref{sec.sec} explains how to account for malicious clients and malicious computation nodes who do not follow the protocol. 
    
    \item \textbf{Secure channels.}
    We assume pairwise secure and authenticated communication channels between the computation nodes. We also assume secure, anonymous communication channels between the clients and computation nodes, in the absence of a shuffler.
\end{itemize}

The final computed statistics may be published to all parties. Given this output and the information within their view throughout the protocol, the computation nodes should only be able to infer information that is bounded by our differential privacy guarantees. In terms of formally proving security, we aim to satisfy the definition of indistinguishable computationally differential privacy (IND-CDP) by Mironov et al.~\cite{mironov09comp}. We discuss the security model in more detail in \Cref{sec.sec}.

\subsection{System Architecture}
As described in the previous subsection, we have a set of clients and a set of computation nodes. First, the set of clients will send their data in some way to the computation nodes. This is called the \emph{data collection phase}. Then, in a \emph{computation phase}, the computation nodes will compute some statistic and publish the output. In our solutions, we introduce an optional dummy generator party who acts similarly to clients and we may assume the existence of an anonymous channel of communication between clients and computation nodes. We discuss alternatives to both of these assumptions in \Cref{sec.dumgen}. The overall flow of data in our system can be seen in \Cref{fig.flow}.

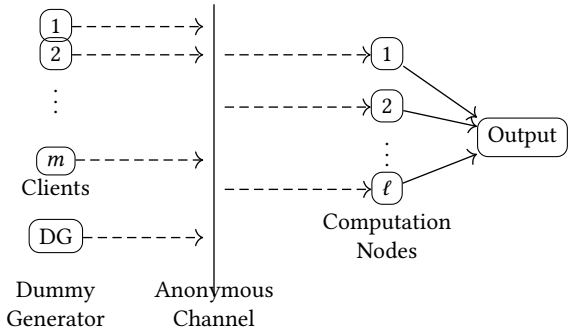
\begin{figure}[ht]
\centering
\begin{tikzcd}[cells={nodes={draw=black,rounded corners,anchor=center,minimum height=1em},column sep =0.1em, row sep=scriptsize}]
1 \arrow[r, dashed]  & |[draw=none]|{} \arrow[ddd,dash,
    start anchor={[yshift=4ex]},
    end anchor={[yshift=-6em]}]   \\[-15pt]
2 \arrow[r, dashed] &  |[draw=none]|{} \arrow[r, dashed] & 1 \arrow[ddr, end anchor=north west]\\[-15pt]
|[draw=none]|\vdots & |[draw=none]| \arrow[r, dashed] & 2 \arrow[dr] \\[-15pt]
m \arrow[r, dashed] & |[draw=none]|{} & |[draw=none]|\vdots & |[yshift=1em,xshift=-6em]|\text{Output} \\[-15pt]
|[draw=none]|\text{Clients} & |[draw=none]|{} \arrow[r, dashed] & \ell \arrow[ur]  \\[-15pt]
\text{DG} \arrow[r, dashed] &|[draw=none]|{} &  |[draw=none]|\text{\begin{tabular}{c}Computation \\ Nodes\end{tabular}} \\[-15pt]
|[draw=none]|\text{\begin{tabular}{c}Dummy \\ Generator\end{tabular}} & |[draw=none]|\text{\begin{tabular}{c}Anonymous \\ Channel\end{tabular}}
\end{tikzcd}
\caption{Data flow}
\label{fig.flow}
\end{figure}

This architecture is similar to systems currently in use that involve a large set of clients sending data to multiple servers, such as Prio~\cite{prio,letsencrypt-prio}. The use of multiple servers for computation prevents the aggregation of private data at one central authority. Recall the example use-case of private telemetry data collection. Clients in this setting represent the web browsers that send data to the servers responsible for the computation. For example, these servers may be controlled by browser vendors and/or certificate authorities~\cite{letsencrypt}.

%% file: Initial_soln.tex
In this section, we provide a high-level overview of some preliminary solutions that fit within the architecture and trust assumption. This ultimately helps us demonstrate how we came to our final solution.

We can distinguish between solutions which conceal keys and those which do not.
Solutions that hide the keys incur a high computational cost to index the records into an array. So, as preliminary solutions, we only consider solutions that do not conceal keys. 
We begin with a naive solution that minimizes trust but requires high communication costs. The solution involves two types of participants: clients who own the data and computation nodes who calculate the statistic. We assume that there exists secure, anonymous channels between the clients and the computation nodes. We call this the \emph{full padding} solution, and it is executed in the following two steps.
\begin{description}\itemsep0em
    \item[Step 1: Padding, Share Generation, and Encryption.] Clients create a \emph{dummy} key-value pair for each key not in their set, $\allkv_i$. For each key-value pair, the client applies a secret sharing scheme to generate $\compnodes$ shares of the value. The clients forward the key-share pairs to the corresponding computation nodes.
    \item[Step 2: Computation.] Upon receiving the data, the computation nodes perform an MPC protocol to compute the differentially private statistic for each key and reconstruct the final result.
\end{description}
Since the computation nodes only see shares of the values, they cannot distinguish between real key-value pairs and dummies. Also, since every client sends a key-value pair, real or fake, for each key, it is impossible for the computation nodes to determine the true frequency of a key.

The obvious disadvantage to this protocol is the communication cost. Each client has $\numkeys$ key-value pairs, for which they generate $\compnodes$ shares, resulting in a communication cost of $O(m \numkeys \compnodes)$. These costs are prohibitively expensive, so we consider methods of improving our approach. We remark that this cost is incurred to prevent the computation nodes from learning who has what data and the true frequency of the keys.

To improve upon this naive solution, we can leverage differentially private padding to mask the true frequency of the keys. This is accomplished by adding a semi-trusted dummy generator.
The protocol, which we refer to as the \emph{one-sided dummy} solution, requires the following three steps.

\begin{description}
\item[Step 1: Share Generation.] For each key-value pair they own, the client applies a secret sharing scheme to the value and generates $\compnodes$ shares. They forward this data to the corresponding computation nodes.
\item[Step 2: Padding.] The dummy generator generates ``enough'' dummy key-value pairs to satisfy differential privacy. They generate shares for these new values via the same secret sharing scheme used by the clients. The data is forwarded to the corresponding computation node.
\item[Step 3: Computation.] The computation nodes receive all of the data and then perform a multi-party computation protocol to compute the desired differentially private statistic. They collectively output the result.
\end{description}

This solution employs differential privacy in two ways. First, to bound the information learned as a result of adding fewer dummies, and second, to output a noisy statistic. In the former, we require a randomized algorithm to add dummies; however, only adding key-value pairs creates a one-sided distribution and thus has a chance of being blatantly non-private. 

There exists work that uses a similar approach of only adding fake records from one-sided distributions in different settings~\cite{mazloom18,djoin,shrinkwrap}. We take inspiration from Mazloom and Gordon \cite[Appendix~A]{mazloom18}, using a shifted and truncated two-sided geometric distribution to sample dummies and achieve approximate DP. In our setting, achieving a reasonable privacy parameter $\delta$ requires adding a large number of dummies. Although it improves on the full padding solution, there remains a high communication cost. We elaborate on the details of this protocol as well as the expected number of dummies and compare it to our final solution in \Cref{app:one_sided_dummy}.

Now that we have established the drawbacks of some preliminary solutions, we can develop a solution which corrects these weaknesses. In our final solution, we reduce the number of dummies and improve the privacy guarantee from approximate to pure differential privacy. Clearly, adding one-sided noise would not allow for pure differential privacy. Intuitively, when only \emph{adding} data and never \emph{removing} data, the output will reveal an upper bound on the true number of data points. 
Thus, we want to simulate the effects of removing without actually removing data, as that would compromise the accuracy. We accomplish this using a technique we call \emph{selective multi-party computation}. This technique involves forwarding each piece of data to a select subset of computation nodes. In doing so, we are able to simultaneously emulate subtractive noise and reduce the communication costs. We describe this in detail in the following section.

%% file: Protocol_desc.tex
This section provides a detailed description of our new protocol designed to compute differentially private statistics over key-value data. Our protocol consists of two phases: data collection and multi-party computation. \change{The main contribution of our work is in the data collection phase. Here, rather than having clients secret share their data among all computation nodes, they only send their shares to a select subset of computation nodes. This atypical approach enables us to improve efficiency and compute differentially private outputs. The multi-party computation phase adapts previously known techniques to account for this change in the data collection phase.}

\subsection{Data Collection Phase}\label{sec.overview}
Recall that we have $m$ clients, each denoted by an index $i \in [m]$ and $\ell \geq 3$ computation nodes. We also require an entity for dummy generation, which we assume is semi-honest. In \Cref{sec.dumgen}, we discuss an alternative approach for dummy generation which does not require a single semi-honest entity. Our protocol aims to return a statistic over the multiset, $\allkv$, while revealing only differentially private information about any particular client's set of key-value pairs $\allkv_i$. In the example of collecting telemetry data, this phase corresponds to the step where usage data is sent from users' browsers to the computation servers.

The data flow is in this phase is similar to the one-sided dummy solution and is summarized in \Cref{fig.flow}.
\begin{enumerate}
    \item Let $\compsub \in \{2,\dots,\compnodes-1\}$ be a (publicly-known) system parameter chosen ahead of time. For each key-value pair belonging to each client $i \in [m]$, the client $i$ chooses a subset of computation nodes of size $\compsub$, uniformly at random, without replacement. The client generates $\compsub$ shares of their value, one for each computation node they selected, using the additive secret sharing scheme. The client also generates $\compsub$ shares of a flag with value 1. This process is summarized in \Cref{fig.client}. The client forwards each tuple consisting of the key, a share of the flag, and a share of the value, to the corresponding computation node. 

    \begin{figure}[h]
    \centering
    \adjustbox{scale=0.8,center}{%
    \begin{tikzcd}[row sep = tiny, cells={nodes={draw=black, rounded corners,anchor=west,minimum height=1em}},
  			]
    		& \text{key, value\_share\textsubscript{1}, flag\_share\textsubscript{1}} \\
    		& \text{key, value\_share\textsubscript{2}, flag\_share\textsubscript{2}} \\
    |[yshift=2em]| \text{key, value, flag} \urar[end anchor = west] \arrow[end anchor = west]{uur}  \drar[end anchor = west] 	
    		& |[yshift=1em,draw=none, xshift=6em]| \vdots  
    		\\
    		& \text{key, value\_share\textsubscript{t}, flag\_share\textsubscript{t}} \\
    		& |[draw=none, xshift=5em]|\text{Share} 
    \end{tikzcd}}
    \caption{Client process for a single key-value pair}
    \label{fig.client}
    \end{figure}
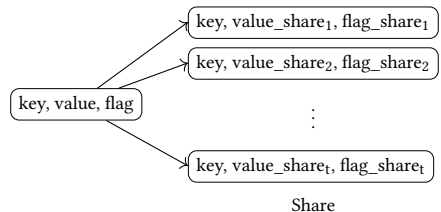

    \item Meanwhile, for each key $k$ in the set $\keyset$, the dummy generator samples $\geonoise$ from a geometric distribution with parameter $\geoparam$. The two-sided geometric distribution is commonly used as the discrete version of the Laplace mechanism~\cite{ghosh2012universally, Balcer_Vadhan_2019}. We use the standard one-sided geometric distribution here to avoid having to shift and truncate the distribution. The dummy generator generates $\geonoise$ \emph{dummy} key-value pairs with key $k$, flag 0, and value 0. They repeat the process performed by the clients by choosing a subset of size $\compsub$ of computation nodes uniformly at random and generating $\compsub$ shares of each dummy value and flag. The dummy generator forwards the dummy data to their respective computation nodes. 
    \item Upon receiving the anonymized data, the computation nodes perform a multi-party computation protocol, based on which statistic is desired, and output the result.
\end{enumerate}
The main difference between this data collection phase and a typical data collection phase in MPC is in step (1) where clients only send their data to a subset of computation nodes. Although conceptually simple, this effectively acts as subtractive noise which will complement the additive noise from the dummy key-value pairs in step (2), thereby ensuring pure differential privacy.

\subsection{Multi-party Computation Phase}\label{sec.ssmpc}
In this section, we describe how the computation nodes process the data once it is received. To compute statistics on the data, they must be able to perform basic operations, such as multiplication and division, and add noise to perturb the final result. A protocol by Wu et al.~\cite{wu16} allows for secure multi-party addition of Laplace noise. We refer readers to the original paper for details of the protocol and proofs of security. Specifically, the protocols allow us to obtain shares of a random variate drawn from a univariate cumulative distribution function.
Additionally, we can use a protocol from Catrina and Saxena to perform division \cite[\S3.4]{catrina2010}. In the following descriptions, we let $[x]$ denote that the value of $x$ is distributed among the parties. We remark that the secret sharing scheme used in the multi-party computation step to generate noise need not be the same additive secret sharing scheme used during data collection; although, a change in schemes may require a \emph{resharing} step to update the shares accordingly.
Finally, we note that due to the selective way in which shares are distributed to computation nodes, the types of statistics which may be computed are restricted to those which can be derived from initially taking the sum of all shares. To compute statistics other than those which we describe below, one might consider encoding the data in various ways \cite{prio} to achieve compatibility with the selective step of the data collection phase.

\subsubsection{Frequency estimation} As per the data collection phase, each computation node receives a subset of shares of key-value pairs. The computation nodes see tuples which contain a key, a share of the flag, and a share of the value. To compute the frequency estimation, the computation nodes compute over the shares of the flag. Recall that real key-value pairs will have a flag with value 1 and dummies will have a flag with value 0 (and this value is not revealed to the computation node since it is secret shared). For each key in $\mathcal{K}$, each computation node executes Protocol \ref{prot.freq} with $\Delta = \lambda$, where $\lambda$ denotes the maximum number of distinct keys held by any client. The value of $\lambda$ is publicly known. Essentially, Protocol \ref{prot.freq} enables the computation nodes to compute the sum of all flags corresponding to some key $k$ and add noise to the result.

\begin{protocol}\caption{Frequency estimation}\label{prot.freq} 
\emph{Input.} A set of triples $\langle k_{ij}, f_{ij}, s_{ij} \rangle$. A target key $k$. A sensitivity $\Delta$. \\
\textbf{1:} Compute the sum $f_i = \sum_{k_{ij} = k} f_{ij}$. These are shares of the frequency $[\freqk]$. \\
\textbf{2:} Generate a random variate, $[\xi]$, drawn from Lap$(\Delta/\freqnoise)$ by invoking Protocol \cite[\S7]{wu16} with $C(t) = \frac{\freqnoise}{2\Delta}\int_{-\infty}^t\exp(-(\freqnoise|s|)/\Delta)ds$. \\
\textbf{3:} Return $[\freqk] + [\xi]$.
\end{protocol}

\paragraph{Correctness.} Before the noise addition, the $i$th computation node has a value $f_i$. The sum of all of these values is 
\[
\sum_{i=1}^\ell f_i = \sum_{i=1}^\ell \sum_{j, k_{ij} = k} f_{ij} = \sum_{k_{ij} = k} f_{ij}.
\]
The sum of all shares of the flags must be the sum of the flags themselves by the properties of the additive secret sharing scheme. As real data points have flag equal to 1 and dummies have flag 0, this must be the frequency of key $k$, $\freqk$. By the properties of Protocol \cite[\S7]{wu16}, the computation nodes ultimately obtain the frequency of the key $k$ with Laplace noise drawn from Lap$(\Delta/\freqnoise)$ where $\Delta=\lambda$.


\subsubsection{Mean estimation} 
To compute a noisy mean estimation, the computation nodes execute Protocol \ref{prot.mean} with $\Delta = \frac{\lambda 2R}{\gamma}$, where $\lambda$ is the maximum number of distinct keys held by any client, $\gamma$ is the minimum frequency of any key, and $2R$ is the size of the interval bounding the values. To summarize Protocol \ref{prot.mean}, the computation nodes compute the frequency of the key $k$, sum the values corresponding to key $k$, divide by the frequency to compute the mean, and then add noise. We let $\sumk$ denote the sum of all values corresponding to key $k$.

\begin{protocol}\caption{Mean estimation}\label{prot.mean}
\emph{Input.} A set of triples $\langle k_{ij}, f_{ij}, s_{ij} \rangle$. A target key $k$. A sensitivity $\Delta$. \\
\textbf{1:} Compute the sum $f_i = \sum_{k_{ij} = k} f_{ij}$. These are shares of the frequency $[\freqk]$.\\
\textbf{2:} Compute the sum $s_i = \sum_{k_{ij} = k} s_{ij}$. These are shares of the sum of values $[\sumk]$.\\
\textbf{3:} Invoke Protocol \cite[\S3.4]{catrina2010} with input $[\sumk],[\freqk]$ to obtain $[\sumk/\freqk]$. \\
\textbf{4:} Generate a random variate, $[\xi]$, drawn from Lap$(\Delta/\meannoise)$ by invoking Protocol \cite[\S7]{wu16} with $C(t) = \frac{\meannoise}{2\Delta}\int_{-\infty}^t\exp(-\frac{\meannoise|s|}{\Delta})ds$. \\
\textbf{5:} Return $[\sumk/\freqk] + [\xi]$.
\end{protocol}

\paragraph{Correctness.} Steps 1-2 accurately compute the frequency of the key $k$ and the sum of all the values corresponding to key $k$, by the properties of the additive secret sharing scheme. This is also due to the fact that the dummies have a value of 0, so this does not affect the sum. Step 3 computes the true mean. By the properties of Protocol \cite[\S7]{wu16}, the computation nodes obtain the mean of the values corresponding to key $k$ with Laplace noise drawn from Lap$(\Delta/\meannoise)$ where $\Delta=\frac{\lambda 2R}{\gamma}$.


\subsection{Implementation Considerations}\label{sec.dumgen}
\paragraph{Anonymous Channels.} In our protocol description, we assume that there exist anonymous channels between the clients and the computation nodes. Although this can be accomplished by using a variety of anonymity systems \cite{tor}, we discuss one approach here. We may rely on a semi-trusted party to act as an intermediary between the clients and the computation nodes. This is similar to the shuffle model in DP \cite{frist_shuffle}. In the shuffle model, an intermediary shuffler performs anonymization, shuffling, thresholding, and batching. The presence of the shuffler induces a privacy amplification which allows each client to add less randomness \cite{google_shuffle}. Our protocol does not require an equally strong trust assumption on the shuffler. Rather, we can have clients encrypt their data with the public keys of the computation node that is responsible for receiving each share of data. Then, the clients forward their data to the shuffler, which removes any user-specific metadata, before forwarding the shares to the respective computation nodes in some random order. This ensures that data is anonymized from the perspective of the computation nodes without revealing the actual data to the shuffler. The only information that the shuffler learns is how many key-value pairs each client holds. Ultimately, our use of this intermediary party requires less trust in the party than the trust required in the shuffle model. Given that the shuffle model is used in practice today by organizations such as Google \cite{google_shuffle}, it is reasonable to assume the existence of anonymous channels.
Note that collusion between the shuffler and the server is not permitted in the shuffle model. Similarly, collusion between the computation nodes and the nodes implementing the anonymous channel is not permitted in our protocol.

In addition to anonymizing data, we may choose the shuffler to be the semi-trusted party responsible for generating the dummy values. Thus, this intermediary party can fulfill two goals at once, while we retain minimal trust assumptions.


\paragraph{Dummy Generation.}
Another requirement in the protocol description is a semi-trusted party to generate dummies.
A single party that generates dummies poses a threat if that party colludes with the computation nodes. If the total number of dummies generated for key $k$ is known to a computation node, they can derive the expected number of dummies it received for said key. For $x_k$ dummies in total, one computation node receives $px_k$ dummies on average. An estimate of the true frequency of the key is derived by deducting $px_k$ from the total number of key-value pairs it receives.

To circumvent this issue, one alternative is using $d$ parties to generate the dummies. Specifically, for each key $k$, $d$ parties independently sample $x_k$ from a geometric distribution with parameter $r$ and generate $x_k$ dummies for key $k$. Note that collusion between dummy generators and clients does not affect the security of the protocol. Hence, for each key, dummy generators can be chosen from amongst the clients to eliminate the requirement for extra entities in the protocol. While the total number of dummies that is generated for each key is roughly $d$ times higher than required, this approach distributes trust and prevents a single point of failure. Moreover, our evaluation shows that the expected number of dummies with a single dummy generator is very small (less than 3 as shown in \Cref{fig:expt_dummies}). Hence, this approach adds a small overhead to the overall protocol.
In \Cref{sec.sec}, we address the case where dummy generators behave maliciously.

\change{\paragraph{Offline Noise Generation} The random noise generated in line 2 of \Cref{prot.freq} and line 4 of \Cref{prot.mean} depend only on the parameters of the protocol, not the data provided by the clients. Hence, these steps, which constitute the bulk of the runtime, can be performed offline, before the clients provide their data. This reduces the perceived latency from the client's perspective.}

%% file: Privacy.tex
\begin{figure}[t]
\centering
		\includegraphics[scale=0.5]{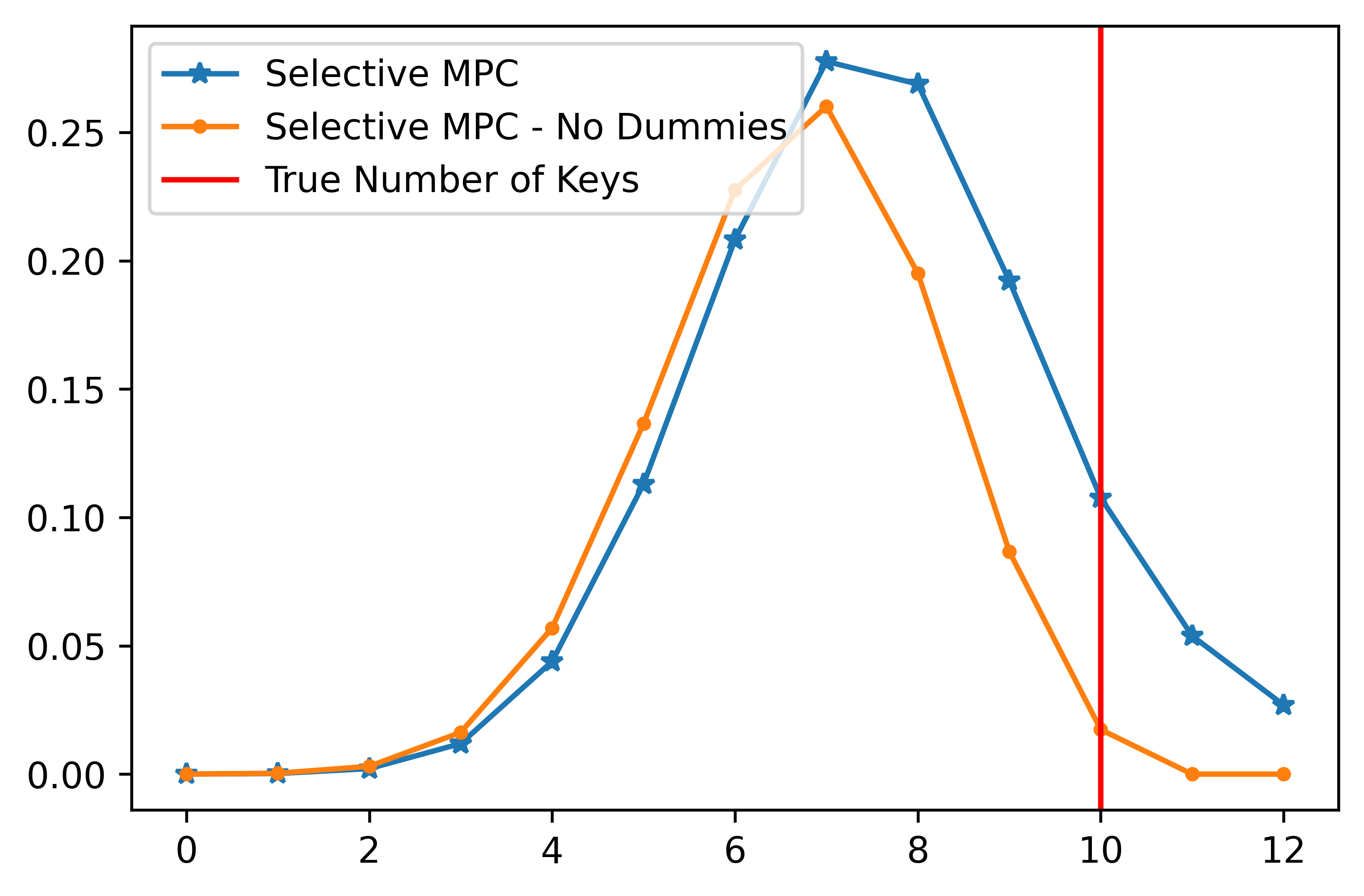}
	    \caption{\change{We visualize the distribution of a single computation node's view of 10 keys. We set $p=2/3$ and $r=0.4$ and plot the distribution with and without dummies.}}
		\label{fig:distribution_vis}
\end{figure}
In this section, we analyze the privacy of our protocols.
First, we formally define the privacy leakage as a function and show that it is $\leaknoise$-DP.
\change{We find that limiting the view of each computation node by randomly selecting which key-value pairs a node sees is equivalent to drawing noise from a binomial distribution. We visualize this distribution in Figure~\ref{fig:distribution_vis} as Selective MPC - No Dummies. In expectation, each node will see $\frac{\compsub}{\compnodes}$ of the keys (2/3 of 10 in the figure), the center of the binomial distribution. Despite being a smooth, two sided distribution, the binomial selection alone is not enough to satisfy pure DP. That is because simply distributing the key-value pairs, ensures it is impossible for a computation node to observe more than the true count of keys (we observe the probability of 11 or 12 keys is zero in Figure~\ref{fig:distribution_vis}). For this reason, in Selective MPC, we add a small number of dummy key-value pairs following a geometric distribution. We can see that this shifts the distribution to the right and avoids clipping the right tail in Figure~\ref{fig:distribution_vis}. We formalize our unique combination of additive noise from a geometric distribution and subtractive noise from a binomial distribution and prove that it satisfies pure DP.}
We show in \Cref{sec:output_privacy} that our frequency estimation satisfies $\freqnoise$-DP and that our mean estimation satisfies $\meannoise$-DP. As a result, the total $\epsilon$ spent for the frequency estimation protocol is $\freqnoise + \leaknoise$ by naive composition. Similarly, the privacy budget for the mean estimation is $\meannoise + \leaknoise$. We remark that both the mean and frequency computation can be executed simultaneously, leading to a total cost of $\freqnoise +\meannoise + \leaknoise$.

Let $\leakage$ represent the algorithm that distributes all the key-value pairs and dummies among the $\compnodes$ computation nodes. Let $\leakage_k$ denote the corresponding algorithm for a single key $k$. For the sake of the privacy analysis, we treat $\leakage$ as an ideal functionality. That is, we assume it is run by a trusted party and consider only the output of the algorithm in our analysis of privacy. An arbitrary output of $\leakage$ is an observation of a single computation node, which we denote by $\observation$. For simplicity, we assume that none of the nodes are colluding with each other and address the case of collusion separately. Recall that upon receiving data, the computation nodes obtain a set of tuples. The shares reveal nothing about the corresponding values and the computation nodes do not know the origin of each piece of data. Therefore, $\observation$ can be described as a histogram of keys observed by the computation node. We wish to show that the output of this algorithm is bounded by differential privacy. The input to $\leakage$ is the list of all key-value pairs held by all users. We say that two inputs, $\allkv$ and $\allkv'$, are \emph{neighbouring} if they differ in at most the addition or deletion of $\distinctkeys$ entries (user-level privacy).

We represent the ground truth of this algorithm as the vector $\freq(\allkv)$, i.e., the true histogram of the keys. Let $\freqk(\allkv)$ represent the count of key $k$ on the input $\allkv$. We assume that each user has at most one value for each key and thus $\forall k \in [\numkeys]$, $|\freqk(\allkv) - \freqk(\allkv')|\leq 1$. The ideal functionality of algorithm $\leakage$ can be described as follows.
\begin{enumerate}
    \item Compute the noisy histogram $Y$ for a given input $\allkv$ by computing $Y_k= \freqk(\allkv) + \eta$ where $\eta\in \mathbb{N}$ is sampled i.i.d from a geometric distribution with parameter $\geoparam$.
    \item Then $Z$ is obtained by sampling from the histogram $Y$, following our protocol. This selective part of our procedure is equivalent to sampling $\|Y\|_1$ times from a binomial distribution with parameter $\subsetratio=\compsub/ \compnodes$.
\end{enumerate}

We denote the probability mass function (PMF) of the binomial distribution\footnote{We follow the standard convention that $\mathbb{B}(z; a,\subsetratio)=0$ whenever $a<z$.} with parameters $a$ and $p$ as
\begin{equation}\label{eqn:pdf_bino}
    \mathbb{B}(z; a,\subsetratio)=\binom{a}{z}p^z(1-p)^{a-z},
\end{equation}
and the PMF of the geometric distribution as
\begin{equation}\label{eqn:pdf_geo}
    \mathbb{G}(z; r)=(1-r)^z r.
\end{equation}
Then, the PMF of $\leakage_k(\allkv)$ is defined as
\begin{equation}
   \Pr[\leakage_k(\allkv)=\observation_k] = \sum_{v=0}^\infty \mathbb{G}(v; \geoparam) \mathbb{B}(\observation_k; \freqk(\allkv)+v,\subsetratio).
\end{equation}
Finally, the PMF of $\leakage(\allkv)$ is 
\begin{equation}
    \Pr[\leakage(\allkv)=\observation]=\prod_{k=0}^\numkeys \Pr[\leakage_k(\allkv)=\observation_k],
\end{equation}
where this equality follows from the fact that $\leakage$ applies randomness independently for each key. Using this PMF we can prove the following theorem.

\begin{restatable}{theorem}{privacythm}\label{thm.priv}
The algorithm $\leakage$ satisfies $\leaknoise$-DP with
\begin{equation}\label{eqn:epsilon_statement}
    \leaknoise = \distinctkeys\ln{\left( \max\left\{ \frac{1}{1-\geoparam},\frac{1}{1-\subsetratio}+1-\geoparam \right\}\right)}.
\end{equation}
\end{restatable}

We defer the proof of \Cref{thm.priv} to \Cref{sec:privacy}. The proof relies on the assumption that computation nodes only see a subset of the keys. This simulates adding noise from a binomial distribution. A natural question to ask is how the privacy guarantee degrades under collusion. In the case where computation nodes collude, the privacy degrades. 
\begin{claim}
  Let $c$ be the collusion threshold which is the maximum number computation nodes that can collude (Theorem~\ref{thm.priv} assumes $c=1$). Then (\ref{eqn:epsilon_statement}) becomes
 \begin{equation}\label{eqn:epsilon_collusion}
     \leaknoise = \distinctkeys\ln{\left( \max\left\{ \frac{1}{1-\geoparam},\frac{\binom{\compnodes}{c+1}}{\binom{\compnodes -c}{c+1}}+1-\geoparam \right\}\right)}.
 \end{equation}
\end{claim}

We note that we consider the minimal value of $\compsub = c + 1$. This is because values of $\compsub$ less than $c+1$ may allow colluding computation nodes to reconstruct the values.

The result in Equation (\ref{eqn:epsilon_collusion}) follows from a similar argument to the proof of \Cref{thm.priv}, substituting $\subsetratio$ with the maximum probability that a node observes a specific share in the presence of collusion. The probability of a computation node observing a single share, either directly or through the nodes it colludes with, is $1 - \binom{\compnodes -c}{c+1}/\binom{\compnodes}{c+1}$.
 In Figure~\ref{fig:collusion} we show how collusion impacts the privacy budget compared to the baseline given by Theorem~\ref{thm.priv}. In the baseline, represented by the dotted line, we set $\compsub=2$, $\compnodes=20$, and choose $\geoparam$ such that it minimizes $\leaknoise$. We plot Equation ($\ref{eqn:epsilon_collusion}$) as a function of $c$ for $\compnodes=20$ and choose $\geoparam$ such that it minimizes $\leaknoise$.
 
\begin{figure}[t]
	\centering
	   \begin{tikzpicture}[]
	   \begin{axis}[xlabel={Collusion threshold $c$},ylabel={$\min\limits_r \epsilon_L$}, ylabel near ticks,width=\columnwidth,height=5cm, legend style={at={(0.5,0.9)}}]
    	    \addplot table [x=c, y=collusion_eps, col sep=comma] {figures/collusion_new.csv};
            \addplot[domain=0.5:8.5, dotted, thick, variable=x] (x, {0.5303});
	   \end{axis}
	   \end{tikzpicture}
	    \caption{$\leaknoise$ for varying levels of collusion.}\label{fig:collusion}
\end{figure}

\change{\paragraph{Outputting Exact Answers.} 
We note that our protocol could be used to output the exact frequency or mean. We would simply follow the selective MPC protocol without adding the final noise (the dummies have no effect on the accuracy). In this case, $\leaknoise$ is the only privacy cost incurred. This is similar to output constrained DP~\cite{he17comp} where the computation nodes do not learn more than a DP amount about the clients data, excluding what they learn from the final result. In general, publishing exact statistics should be avoided as it allows for reconstruction attacks.}


%% file: Security.tex
We start by considering the semi-honest model and later extend to the malicious model. Our protocol consists of two main parts: the data collection and the multi-party computation protocols. For simplicity of the implementation, we are working under the assumption that there exist anonymous, secure channels between the clients and the computation nodes. The clients only see their own data and the final output of the differentially private computation. Our goal is that an adversary learns no more than the differentially private leakage from algorithm $\leakage$ and the final output of the differentially private computation. 

\subsection{Proof of security}
To analyze the security of the protocol, we rely on a definition of indistinguishable computationally differential privacy (IND-CDP) by Mironov et al.~\cite{mironov09comp}. The following definition is for two-party computation. The security definitions which we present encompass the entire protocol, i.e. both the data collection and computation phases.

\begin{definition}[IND-CDP-2PC \cite{he17comp}]
A two-party protocol $\Pi$ for computing function $f$ satisfies $(\epsilon_A(\lambda),\epsilon_B(\lambda)$-indistinguishable computationally differential privacy (IND-CDP-2PC) if $\textsc{VIEW}_A^{\Pi}(D_A,\cdot)$ satisfies $\epsilon_B(\lambda)$-IND-CDP, i.e., for any probabilistic polynomial time (in $\lambda$) adversary $\mathcal{A}$, for any neighbouring datasets $(D_B,D_B')$,
\begin{align*}
&\Pr[\mathcal{A}(\textsc{view}_A^{\Pi}(D_A,D_B))=1] \\
\leq &\exp(\epsilon_B) \cdot \Pr[\mathcal{A}(\textsc{view}_A^{\Pi}(D_A,D_B'))=1] + negl(\lambda).
\end{align*}
Likewise for B's view of any neighbours $(D_A,D_A')$ and $\epsilon_A$.
\end{definition}

For our purposes, we extend the definition from two-party to multi-party computation and we let $\epsilon = \epsilon_A = \epsilon_B$ for convenience. Intuitively, the following definition ensures that an adversary $\mathcal{A}$ cannot distinguish neighbouring databases $D, D'$ from their view, and the same applies to every other party. 

\begin{definition}[IND-CDP-MPC]
A multi-party protocol $\Pi$ for computing function $f$ satisfies $\epsilon(\lambda)$- indistinguishable computationally differential privacy (IND-CDP-MPC) if for every probabilistic polynomial time (in $\lambda$) adversary $\mathcal{A}$ with input $D_A$, and for neighbouring datasets $D,D'$ belonging to the honest parties (i.e. $D,D' = \cup_{i \setminus A} D_i$),
\begin{align*}
&\Pr[\mathcal{A}(\textsc{view}_A^{\Pi}(D_A, D))=1] \\
\leq &\exp(\epsilon) \cdot \Pr[\mathcal{A}(\textsc{view}_A^{\Pi}(D_A, D'))=1] + negl(\lambda).
\end{align*}
Likewise for every other party's view of neighbours $(D,D')$ and $\epsilon$.
\end{definition}


Let $\Pi$ denote the complete protocol, from the clients sending data to the computation nodes and the computation nodes performing the multi-party computation protocol to compute a function $f$ over the data. Our security property will apply to multi-party computation protocols which satisfy \emph{statistical privacy}. We say that the MPC protocol computing $f$ ensures \emph{statistical privacy} if the views of the computation nodes can be simulated by an oracle such that the true and simulated distributions are statistically indistinguishable. Both our protocol computing the frequency estimation and the protocol computing the mean estimation satisfy statistical privacy. Wu et al.~\cite{wu16} prove that their noise generation protocol satisfies statistical privacy and Catrina and Saxena \cite{catrina2010} prove that their division protocol satisfies statistical privacy. The remaining steps in our protocols are secure in an information-theoretic sense. Thus, the following security proof can be applied to both our frequency and mean estimation protocols. 

In the definition that we aim to satisfy, the adversary controls some set of parties, denoted by $A$. The adversary we consider may control at most $b_{N}$ computation nodes, where $b_{N}$ will be defined depending on the scenario. Then, their view consists of the combined view of the $b_{N}$ computation nodes which they control. Note also that an adversary may control up to $b_{C} \leq m-1$ clients (all but one). Since a client does not perform any computation and simply submits data, their view is trivial to simulate. Further, we assume for the time being that the protocol is using a trusted third party to generate dummies, which cannot be controlled by the adversary. In order to determine $b_{N}$, the bound on the number of computation nodes that the adversary may control, we must consider how the MPC protocol functions. In our case, we use a $(t,t)$-threshold scheme, so we can let $b_{N} = t-1$. Then, the neighbouring datasets, $D$ and $D'$, correspond to the data sent by the honest clients which are outside of the control of the adversary. 
Finally, we state our security theorem and defer the proof to \Cref{sec:proof-security}.

\begin{restatable}{theorem}{securitythm}
\label{thm:security}
If $f$ satisfies $\epsilon_f$-differential privacy and the corresponding MPC protocol satisfies statistical privacy, then our protocol $\Pi$ ensures IND-CDP-MPC in the presence of a semi-honest adversary who controls at most $b_{N}$ computation nodes and at most $b_{C}$ clients.
\end{restatable}

\subsection{Extending to the malicious model.}\label{sec.malicious}
Our security proof assumed that we work in the semi-honest model, in which all parties follow the protocol. We now consider how our protocol may function if this assumption is relaxed. First, consider the possibility of malicious computation nodes. In the data collection phase, we require that the computation nodes act semi-honestly so that they are unable to arbitrarily add their own input. We can relax this assumption in the multi-party computation phase. Methods for extending MPC protocols to be secure in the presence of malicious parties have been discussed in the literature~\cite{spdz, mp-spdz}. We can apply the same strategies to our own work, using protocols that are maliciously secure under the assumption that there exists an honest majority \cite{bgw88,dwork06, spdz}. 
This prevents nodes from adding arbitrary noise during the computation phase which may affect the results. In this scenario, we can adjust our bound $b_{N}$ on the number of computation nodes that the adversary may control to be the minimum of $t-1$ and $\lfloor(\ell-1)/2\rfloor$. This ensures that the adversary cannot control a majority of participants in the multi-party computation. \change{As with any other multi-party computation protocol which assumes an honest majority, if the assumption of an honest majority is not satisfied then there are no guarantees on the privacy or correctness of the output. In this paper, we have discussed how our protocol can account for semi-honest computation nodes and a minority of misbehaving computation nodes; however, we are not limited to these two settings. Since the multi-party computation phase is a simple adaption of previously known techniques in MPC, it is straightforward to substitute an MPC protocol with other guarantees. For example, in a scenario where one is concerned about the reconstruction of invalid or incorrect results, we can use an MPC protocol which ensures robustness. Similarly, we could use a protocol which is secure against Byzantine failures if we are concerned about computation nodes failing. The protocol which is ultimately chosen depends on the requirements of the system in practice, but is not limited in any way by our construction.}

Next, we consider the possibility of malicious clients. We can prevent malicious clients from submitting invalid data by having the computation nodes perform validation checks on the shares they receive. \change{Note that the validity check is conducted in MPC by the $t$ computation nodes that receive each piece of data, not all $\ell$ computation nodes.
For example,} when $t$ computation nodes receive the shares of some flag $x$, they
can compute $x\cdot(1-x)$ to check that the result is 0 (which we evaluate in \Cref{app.input_val_exp}). This verifies that the flag is in the set $\{0,1\}$. The computation nodes learn no additional information by computing validation checks with the other nodes which receive shares of the same data. Similar checks can be computed on the data received from the clients to ensure that it exists in some valid range. Thus, we are able to ensure that clients are submitting valid data within some pre-defined, acceptable range.
\change{Validity checks other than those described here can also be implemented and our protocol does not impose any restriction on input validation.}
These verification checks can be implemented in the setting where there is a dishonest majority among the subset of computation nodes which receive shares by using MPC protocols like SPDZ \cite{SPDZ20}.

Finally, suppose that the third party generating dummies is malicious. \change{On one hand, a malicious dummy generator may submit additional data points to attempt to influence the final result. This poses no more of a threat than the case where malicious clients submit false data points. This is a potential vulnerability to any system which collects data from clients and is not prevented by security against malicious adversaries. In such a scenario where the dummy generator is malicious, it is also possible that they submit no dummies or submit an incorrect number of dummies. This would compromise the differential privacy guarantees on the leakage of our protocol, but would not compromise the correctness of the output.} It is difficult to ensure that dummy generation by a single party occurs correctly. However, recall that we have discussed alternatives to dummy generation which provide more distribution of trust in \Cref{sec.dumgen}. Suppose we are using the method described in \Cref{sec.dumgen} where several clients are responsible for generating dummies. Then, we must ensure that the adversary cannot control all of the clients generating dummies. If we assume that the adversary can control at most $b_{C}$ clients, then we can let $b_{C}+1$ clients generate dummies, which ensures that we still achieve differential privacy. This does not affect the security proof because the adversary is unable to control at least one party generating dummies. Ultimately, distribution of trust allows us to consider the presence of malicious parties while achieving secure dummy generation.

%% file: Evaluation.tex
This section provides a theoretical evaluation of our protocol in terms of accuracy, communication, and computation cost. We then compare this with other methods for computing statistics on key-value data. Finally, we present our experimental results to demonstrate the practicality of our protocol.

\subsection{Our Protocol} \label{sec:eval-our-prot}

For the purpose of comparison, we analyze our protocol in the case where $t=2$. This only affects the communication cost during the data collection phase and is a reasonable choice in practice as it minimizes $\leaknoise$.

\paragraph{Communication.} The communication conducted in our protocol can be decomposed into the communication from the clients (and dummy generator) to the computation nodes, i.e.~data collection, and the communication between computation nodes. 
For comparison, we focus mainly on client-to-server communication. The communication between computation nodes is specific to solutions using MPC and depends largely on parameters of the MPC protocol, such as the precision of the noise generated. These are chosen by a system administrator depending on the context of the situation. We provide some experimental results for communication between computation nodes in \Cref{sec:experiments}. In this section, we elaborate on the communication between clients and servers.
First, $2\numkv$ elements are sent from the clients to the computation nodes because each client generates $\compsub=2$ shares for each of the key-value pairs it owns. The number of dummies added for each key follows a geometric distribution with an expected value of $\frac{1}{\geoparam}$, which means $\frac{\numkeys}{\geoparam}$ dummies are generated on average. For each dummy, $\compsub=2$ shares are generated. Therefore, in expectation, $\theta(\numkv + \frac{\numkeys}{\geoparam})$ shares are sent from the clients (and dummy generator) to the computation nodes.

\paragraph{Computation.} We define the computation cost of our protocol to be the number of multiplications in the MPC protocol. For each key, the noise generation (and division in the case of mean estimation) is run to derive the result. The complexity of these operations depends on the chosen protocol and parameters related to their implementation, such as precision. Hence, we consider the complexity to be constant with respect to each key. Therefore, the overall number of multiplications in the MPC protocol across all keys is $O(n)$.

\paragraph{Accuracy.}
Finally, we consider the accuracy of the output of our protocol. Our protocol implements an algorithm with output identical to that of the central model, so it suffices to consider the accuracy of the Laplace mechanism, which is well known. We begin by evaluating the accuracy of our frequency estimation protocol denoted by $F(\allkv)$. Using a tail bound on the Laplace distribution along with the union bound as shown in Theorem 3.8 of Dwork and Roth~\cite{dwork14}, we obtain the following accuracy guarantee $\forall \beta \in (0,1]$.
\begin{equation}
    \Pr[||\freq(\allkv) - \freqalg(\allkv)||_\infty \geq \ln(\numkeys/\beta)\cdot (\distinctkeys/\freqnoise)] \leq \beta.
\end{equation}

Similarly, for the mean estimation protocol, with output denoted $\meanalg(\allkv)$, we obtain that $\forall \beta \in (0,1]$,
\begin{equation}
    \Pr[||\mean(\allkv) - \meanalg(\allkv)||_\infty \geq \ln(\numkeys/\beta)\cdot (2R\lambda/\meannoise \minfreq)] \leq \beta.
\end{equation}

To allow for comparison of different methods, we simplify the above expression and obtain an approximation for the error of our protocols using the standard deviation. For a single key, we compute that the error is approximately $O\left(\Delta/\epsilon\right)$ where $\Delta$ is the sensitivity of the specific protocol. For simplicity, we will only consider the error for a single key's frequency or mean and note the case of all keys can easily be obtained using a union bound. For our analysis, we assume all clients hold this key. Taking into account the privacy budget from the leakage function, we conclude that our solution achieves an approximate error bound of $O(\Delta/(\epsilon-\leaknoise))$. 

\begin{table}[t]
\addtolength{\tabcolsep}{-3pt}
\begin{tabular}{@{}ccccc@{}}
\toprule
             & \begin{tabular}[c]{@{}c@{}}Client-to-Server\\ Communication\end{tabular}  & \begin{tabular}[c]{@{}c@{}}MPC\\ Complexity\end{tabular} & Error \\ \midrule
Our Protocol &  $\theta(\numkv + \frac{\numkeys}{\geoparam})$ & $O(\numkeys)$ & $O\left( \Delta/(\epsilon-\leaknoise)\right)$\\
Central DP   &  $O(\numkv)$  &     -         & $O\left( \Delta/\epsilon\right)$ \\
Naive MPC    &  $O(\numkv \ell)$ & $O(\numkv \numkeys)$ & $O\left(\Delta/\epsilon\right)$ \\
PrivKVM       &    $O(m\log \numkeys)^*$    &     -    & $O\left(\sqrt{m}\Delta/\epsilon\right)^*$ \\
PCKV         &  $O(m\log n)$  &     -    &  $O\left(\sqrt{m}\Delta/\epsilon\right)$ \\ \bottomrule
\end{tabular}
\addtolength{\tabcolsep}{3pt}
\caption{Comparison of client-server communication, computation (number of multiplications), and error (approximate bound). *Includes a multiplicative factor equal to the number of iterations which we exclude.}
\label{tab:eval}
\end{table}

\begin{figure*}
\begin{minipage}{0.45\textwidth}
    \centering
	\begin{tikzpicture}
	 \begin{axis}[xlabel=$\ell$,ylabel={Time (seconds)},width=0.95\textwidth,ylabel near ticks,xlabel near ticks,height=5.5cm,ymin=0]
  
	 \addplot [blue] table [x=nodes, y=mean, col sep=comma] {figures/local_time_freq.csv}; \label{time_freq}
	 \addplot [name path=upper,draw=none] table[x=nodes, y=upper, col sep=comma] {figures/local_time_freq.csv};
	\addplot [name path=lower,draw=none] table[x=nodes, y=lower, col sep=comma] {figures/local_time_freq.csv};
	\addplot [fill=blue!10] fill between[of=upper and lower];

	 \addplot [red] table [x=nodes, y=mean, col sep=comma] {figures/local_time_mean.csv}; \label{time_mean}
	 \addplot [name path=uppermn,draw=none] table[x=nodes, y=upper, col sep=comma] {figures/local_time_mean.csv};
	\addplot [name path=lowermn,draw=none] table[x=nodes, y=lower, col sep=comma] {figures/local_time_mean.csv};
	\addplot [fill=red!10] fill between[of=uppermn and lowermn];

	 \addplot [green] table [x=nodes, y=mean, col sep=comma] {figures/local_time_mean_mal.csv}; \label{time_mean_mal}
	 \addplot [name path=uppermn,draw=none] table[x=nodes, y=upper, col sep=comma] {figures/local_time_mean_mal.csv};
	\addplot [name path=lowermn,draw=none] table[x=nodes, y=lower, col sep=comma] {figures/local_time_mean_mal.csv};
	\addplot [fill=green!10] fill between[of=uppermn and lowermn];

\end{axis}
	 
	 \begin{axis}[width=0.95\textwidth,axis y line*=right,axis x line=none,ylabel=$\epsilon_L$,ylabel near ticks,legend style={font=\small,at={(0.7,0.98)}},height=5.5cm,ymin=0]
	 \addlegendimage{/pgfplots/refstyle=time_freq}\addlegendentry{Freq. Est.}
	\addlegendimage{/pgfplots/refstyle=time_mean}\addlegendentry{Mean Est.}
 	\addlegendimage{/pgfplots/refstyle=time_mean_mal}\addlegendentry{Mean Est. (Mal.)}
	 \addplot [orange] table [x=nodes, y=epsilon, col sep=comma] {figures/local_time_freq.csv};
	 \addlegendentry{$\leaknoise$};
	\end{axis}
	\end{tikzpicture}
        \captionof{figure}{Run time for a single key. Measurements are means with 95\% confidence interval in shaded region. $\leaknoise$ is shown for the number of computation nodes.}
        \label{fig:local_time}
\end{minipage}~%
\hspace{0.5cm}
\begin{minipage}{0.45\textwidth}
\begin{tikzpicture}

	 \begin{axis}[xlabel=$\ell$,ylabel={Data sent (MB)},width=0.95\textwidth,ylabel near ticks,xlabel near ticks,height=5.5cm,ymin=0]
	 \addplot [blue] table [x=nodes, y=mean, col sep=comma] {figures/local_comm_freq.csv}; \label{comm_freq}
	 \addplot [name path=upper,draw=none] table[x=nodes, y=upper, col sep=comma] {figures/local_comm_freq.csv};
	\addplot [name path=lower,draw=none] table[x=nodes, y=lower, col sep=comma] {figures/local_comm_freq.csv};
	\addplot [fill=blue!10] fill between[of=upper and lower];

	 \addplot [red] table [x=nodes, y=mean, col sep=comma] {figures/local_comm_mean.csv}; \label{comm_mean}
	 \addplot [name path=uppermn,draw=none] table[x=nodes, y=upper, col sep=comma] {figures/local_comm_mean.csv};
	\addplot [name path=lowermn,draw=none] table[x=nodes, y=lower, col sep=comma] {figures/local_comm_mean.csv};
	\addplot [fill=red!10] fill between[of=uppermn and lowermn];

	 \addplot [green] table [x=nodes, y=mean, col sep=comma] {figures/local_comm_mean_mal.csv}; \label{comm_mean_mal}
	 \addplot [name path=uppermn,draw=none] table[x=nodes, y=upper, col sep=comma] {figures/local_comm_mean_mal.csv};
	\addplot [name path=lowermn,draw=none] table[x=nodes, y=lower, col sep=comma] {figures/local_comm_mean_mal.csv};
	\addplot [fill=green!10] fill between[of=uppermn and lowermn];

\end{axis}
	 
	 \begin{axis}[width=0.95\textwidth,axis y line*=right,axis x line=none,ylabel=$\epsilon_L$,ylabel near ticks,legend style={font=\small,at={(0.7,0.98)}},height=5.5cm,ymin=0]
	 \addlegendimage{/pgfplots/refstyle=comm_freq}\addlegendentry{Freq. Est.}
	\addlegendimage{/pgfplots/refstyle=comm_mean}\addlegendentry{Mean Est.}
	\addlegendimage{/pgfplots/refstyle=comm_mean_mal}\addlegendentry{Mean Est. (Mal.)}
	 \addplot [orange] table [x=nodes, y=epsilon, col sep=comma] {figures/local_comm_freq.csv};
	 \addlegendentry{$\leaknoise$};
	\end{axis}
	\end{tikzpicture}
        \captionof{figure}{Communication for a single key. Measurements are means with 95\% confidence interval in shaded region. $\leaknoise$ is shown for the number of computation nodes.}
        \label{fig:local_data}
\end{minipage}
\end{figure*}

\subsection{Comparison with Other Work}
We now compare the communication, computation, and accuracy of our protocol with a naive MPC solution that conceals and secret shares keys, existing local DP solutions (PrivKVM \cite{ye19}, PCKV \cite{gu2019pckv}), and the central model. The results of our comparison are summarized in \Cref{tab:eval}. Although there has been no prior work on key-value data in the shuffle model, we discuss potential solutions in this setting in \Cref{sec.rel}.

\paragraph{Naive MPC} A naive application of MPC would have each client secret share keys and values and forward the shares to each computation node. This requires no dummies; however, the computation nodes cannot determine which shares correspond to which key. Each time they want to perform a computation over a key-value pair, they need to collectively index into a table containing all the keys. This results in an additional $\numkv\numkeys$ multiplications in the MPC protocol. Since clients create a share for each computation node, the client-to-server communication cost is $O(\numkv\ell)$. This approach simulates the central model, so it has an approximate error of $O(\Delta/\epsilon)$.

\paragraph{PrivKVM} PrivKVM uses a local DP protocol in which each client sends a single perturbed key-value pair to the data curator. Each client samples this key-value pair randomly from their set of key-value pairs. This constitutes all the communication in the protocol. Thus, using our notation for the size of the key domain and the number of users, the total communication complexity of this protocol is $O(m\log n)$. The curator calculates the desired statistics (frequency and mean) by iterating over the data they received and then calibrating the result. This computation cost is negligible in comparison to an MPC protocol. Finally, PrivKVM utilizes the Harmony protocol~\cite{nguyn2016collecting} and thus achieves an approximate error bound of $O(\sqrt{m}\Delta/\epsilon)$. We remark that the accuracy depends on the number of iterations $c$ carried out in the PrivKVM protocol. However, we follow the approach taken by the authors and treat $c$ as a constant in our analysis. Additionally, there would be a sampling error due to the fact each client only sends one of their keys pairs which we omit from this analysis.

\paragraph{PCKV}
This protocol improves upon PrivKVM. To estimate the mean, each client sends either a key-value pair or a vector equal to the size of the key domain. For our comparison, we will consider the former which has a lower overall communication complexity of $O(m\log n)$ across all users. The protocol does not require interaction, so there is no multiplicative factor for the number of rounds. In terms of accuracy, all of the improvements influence the error by a constant amount. While this showed significant improvements in experimental evaluations, it has no asymptotic advantage over PrivKVM.

\subsection{Experimental Results}\label{sec:experiments}

To evaluate the practicality of our protocol, we compute a series of benchmarks. We recall that our protocol consists of two phases: the data collection and multi-party computation phase. The bottleneck of the data collection phase is the communication of data from all clients. Collecting shares is as simple as sending a single message (to several nodes). This overhead is shared by all other related work and a detailed comparison of the size of the communication is given in Table~\ref{tab:eval}. Moreover, the bulk of the overhead in the data collection phase consists of the time required for coordination amongst the clients, i.e. the time for all clients to send their data, which is not reflected in experiments. The computational overhead of our system in the data collection phase is the dummy generation which incurs a negligible overhead equivalent to sharing a few additional key-value pairs (we evaluate exactly how many in \Cref{app:compare}). The most expensive part of our protocol is the multi-party computation phase which requires tightly synchronized nodes. Thus, we implement this phase using the noise generation protocol from Wu et al.~\cite{wu16} and provide benchmarks for the run time and communication. We investigate the effects of the number of computation nodes, network latency, and the number of keys.
Our implementation and scripts for the evaluation are publicly available.~\footnote{https://git.uwaterloo.ca/r5akhava/selective-mpc}

\paragraph{Setup}
We use the MP-SPDZ~\cite{mp-spdz} framework to implement Protocol \cite[\S7]{wu16} from Wu et al.~for generating a Laplace random variate.
We simulate the frequency or mean estimation protocol by sampling from the Laplace distribution and adding it to an arbitrary result (either a constant or a division of constants).
We use the Shamir protocol, implemented in the MP-SPDZ framework.~\footnote{https://github.com/data61/MP-SPDZ}
Using this implementation, we observed that the underlying technique for generating Beaver triples switches when the number of computation nodes exceeds a threshold.
The default threshold was set too high for our use case, resulting in sub-optimal run time results.
We set the threshold to ten computation nodes to optimize run time results; however, there remains a noticeable change in the communication results.  

\paragraph{Number of Comp. Nodes}
First, we measure how our protocol scales with the number of computation nodes. Since our protocol can be parallelized over the number of keys, we only consider a single key for this benchmark. To create multiple parties, we use up to 30 Amazon EC2 T2 instances. Each machine has four cores, 16 GB of memory, and is located in the same Ohio data center.

In Figure~\ref{fig:local_time}, we plot the average run time over 100 runs with 95\% confidence intervals as the number of nodes increases for frequency and mean estimation \change{in the semi-honest model and mean estimation in the malicious model.} For context, we also plot the corresponding $\leaknoise$ for each $\compnodes$. We see roughly quadratic growth in the run time, which is consistent with the asymptotic complexity of MPC. \change{In the mean estimation protocol, for 30 nodes (where $\leaknoise = 0.512$) we observe 0.88 and 6.36 seconds for the semi-honest and maliciously secure protocol, respectively.
Moreover, for 5 nodes (where $\leaknoise=0.758$) the protocol runs in 0.07 and 0.83 seconds in the semi-honest and malicious model, respectively.}
Figure~\ref{fig:local_data} shows the corresponding total communication. The dip around ten nodes comes from the change in protocol mentioned in the previous section. \change{In the mean estimation protocol, we observe 229.45 MB and 2976 MB for 30 nodes in the semi-honest and malicious model and 8.77 MB and 36.11 MB for five nodes in the semi-honest and malicious model.} We conclude that our protocol scales acceptably to realistic numbers of computation nodes.

\paragraph{Different Data Centers}
Having all computation nodes in the same data center is not always possible in practice. Thus, in this experiment, we consider the effects of network latency from having computation nodes in different locations. To do this, we consider three different scenarios: local, remote, and distant. In the local scenario, all nodes are in the Ohio data center. In the remote setting, half of the nodes are in the Ohio data center and half in Northern California. Finally, the distant scenario has half of the nodes in the Ohio data center and the other half in Frankfurt, Germany. We consider a single key and conduct this experiment for 3, 6, 10, and 20 computation nodes. \change{Moreover, we experiment with a semi-honest and malicious secure MPC protocol.}
The results are given for 100 runs in \Cref{tab:nodes_time}. We observe reasonable run times even for a high number of computation nodes. Despite increased latency in wide area networks, the total run time remains under six seconds. 

\begin{figure}
    \centering
    \addtolength{\tabcolsep}{-2pt}
    \begin{tabular}{ccccc}
        \toprule
        \# Nodes & $\leaknoise$ & Local & Remote & Distant \\
        \midrule
        \multirow{2}{*}{3} & \multirow{2}{*}{1.19} & 0.037~(0.002) & 2.323~(0.020) & 4.50~(0.52)\\ 
        & & 0.19~(0.04) & 2.91~(0.06) & 5.66~(0.12) \\ \midrule
        
        \multirow{2}{*}{6} & \multirow{2}{*}{0.69} & 0.077~(0.004) & 2.65~(0.15) & 4.80~(0.12)\\ 
        & & 0.59~(0.03) & 3.55~(0.11) & 6.38~(0.14) \\ \midrule
        
        \multirow{2}{*}{10} & \multirow{2}{*}{0.59} & 0.237~(0.047) & 3.01~(0.12) & 5.85~(0.27)\\ 
        & & 1.12~(0.05) & 4.34~(0.18) & 7.53~(0.24) \\ \midrule

        \multirow{2}{*}{20} & \multirow{2}{*}{0.53} & 0.573~(0.064) & 3.18~(0.14) & 5.73~(0.22) \\
        & & 2.91~(0.09) & 6.68~(0.27) & 10.17~(0.60) \\
        \bottomrule
\end{tabular}
\addtolength{\tabcolsep}{2pt}
\captionof{table}{Run time (in seconds) for one key. Measurements are means with standard deviations in parentheses. \change{For each configuration the first row and second row use a semi-honest and maliciously secure MPC, respectively.}}
\label{tab:nodes_time}
\vspace{1em}
\begin{tabular}{cccc}
        \toprule
        \# Keys & Local & Remote & Distant \\
        \midrule
        \multirow{2}{*}{10} & 0.075~(0.003) & 2.58~(0.09) & 5.015~(0.08)\\ 
        & 0.15~(0.00) & 3.10~(0.07) & 6.08~(0.43)\\ \midrule

        \multirow{2}{*}{100} & 0.58~(0.03) & 5.0~(0.3) & 9.4~(0.4)\\ 
        & 1.34~(0.11) & 5.64~(0.37) & 10.64~(0.27)\\ \midrule

        \multirow{2}{*}{1000}  & 2.38~(0.03) & 32.52~(0.87) & 63.6~(1.1)\\ 
        & 3.59~(0.06) & 34.28~(0.94) & 65.86~(0.38)\\ \midrule

        \multirow{2}{*}{10000}  & 20.00~(0.07) & 307.98~(8.09) & 602.3~(5.6)\\ 
        & 26.64~(0.09) & 323.65~(0.23) & 627.30~(0.39)\\
        \bottomrule
    \end{tabular}
    \captionof{table}{Run time (in seconds) for many keys with 5 computation nodes, which corresponds to $\leaknoise=0.758$. Measurements are means with standard deviations in parentheses. \change{For each configuration the first and second row use a semi-honest and maliciously secure MPC, respectively.}}
    \label{tab:keys_time}
\end{figure}
\paragraph{Varying Key Domain Sizes}
Furthermore, we investigate the effects of multiple keys. Recall that each key is treated separately in our protocol, so we parallelize the computation over the keys. To do this, we use the built-in threading functionality in MP-SPDZ. We use five computation nodes (which achieves $\leaknoise=0.758$), each running an Amazon EC2 C5 instance with 64 cores and 128 GB of memory. We consider key domain sizes of 10, 100, 1000, and 10,000. As with the previous experiment, we use local, remote, and distant network setups. Results are given for 40 runs in \Cref{tab:keys_time}. We observe reasonable run times that increase approximately linearly after the number of keys exceeds the number of cores.

\paragraph{Input Validation}
We implement and examine the cost of validating the client's input in \Cref{app.input_val_exp}. Our experiments suggest that input validation does not impose an impractical burden on the overall protocol.

%% file: Discussion.tex
\paragraph{Privacy Limitations}

An interesting property of our protocol is that for a fixed choice of $\compnodes$, there is a lower bound on the leakage privacy $\leaknoise$. This differs from other common DP mechanisms where the privacy budget can be made arbitrarily small, albeit with detrimental effects on the utility. We assume $\compnodes$ has a constant value because it is typically a fixed aspect of the setup. In contrast, $\geoparam$ is easily modifiable so we can choose $\geoparam \in (0,1)$ to minimize $\leaknoise$, where $\geoparam$ is the parameter for the geometric distribution.
Figure \ref{fig:privacy-r} plots the leakage privacy as a function of $\geoparam$ for three fixed values of $\compnodes$ based on the result of \Cref{thm.priv}. The lower bound on $\leaknoise$, which arises from the $\max$ function in the formula for $\leaknoise$, can be seen in the graphs. Specifically, for a given $\compnodes$, the smallest achievable $\leaknoise$ can be derived as follows.

\begin{align}
    \min \epsilon_L 
    =& \min_{r,t} \left\{ \Delta\ln{\left( \max\left\{ \frac{1}{1-\geoparam},\frac{1}{1-\frac{\compsub}{\compnodes}}+1-\geoparam \right\}\right)} \right\} \\
    =& \Delta \ln \left(\frac{2(1-2/\compnodes)}{\sqrt{1+4(1-2/\compnodes)^2}-1}\right)
    \label{eq:achievable}
\end{align}

which is achieved at $t=2$ and 
\begin{align}
    r = 1 - \frac{\sqrt{1+4(1-2/\compnodes)^2}-1}{2(1-2/\compnodes)}    
\end{align}

Figure \ref{fig:achievable} graphs the smallest achievable privacy of the leakage function based on Equation (\ref{eq:achievable}). The minimum of Equation (\ref{eq:achievable}) is obtained at $\compnodes=\infty$, albeit impossible to achieve in practice, which gives a lower bound on $\epsilon_L$. In particular, if $\Delta=1$ for simplicity, then

\begin{align}
    \epsilon_L > \min_{\compnodes} \ln \left(\frac{2(1-2/\compnodes)}{\sqrt{1+4(1-\frac{2}{\compnodes})^2}-1}\right) = \ln\left(\frac{2}{\sqrt{5}-1}\right) = 0.4812.
\end{align}

Such a lower bound may be unacceptable in certain scenarios. 
Tightening the analysis of \Cref{thm.priv} would improve upon this bound.

\begin{figure}[!ht]
    \centering
    \begin{subfigure}{.45\columnwidth}
        \begin{tikzpicture}[]
            \begin{axis}[xmin=0.2, ymin=0.4, samples=50, xlabel=$r$, ylabel=$\epsilon_L$,ylabel near ticks,width=4.5cm,height=5cm]
            \addplot[blue, thick, domain=0.1:0.7, variable=\r] (\r,{ln(max(1/(1-0.4)+1-\r,0.05 + 1/(1-\r)))});
              \addlegendentry{$\compnodes=5$}
            \addplot[orange, thick, domain=0.05:0.7, variable=\r] (\r,{ln(max(1/(1-0.2)+1-\r,0.03+1/(1-\r)))});
              \addlegendentry{$\compnodes=10$}
            \addplot[red, thick, domain=0.05:0.7, variable=\r] (\r,{ln(max(1/(1-0.1)+1-\r,1/(1-\r)))});
              \addlegendentry{$\compnodes=20$}
            \end{axis}
        \end{tikzpicture}
        \caption{}
        \label{fig:privacy-r}
    \end{subfigure}\hspace{5mm}%
    \begin{subfigure}{.45\columnwidth}
        \begin{tikzpicture}[]
            \begin{axis} [xmin=0, ymin=0.4, samples=35, xlabel=$\compnodes$, ylabel=$\min\ \epsilon_L$,ylabel near ticks,width=4.5cm,height=5cm]
                  \addplot[domain=3:50, thick, variable=\l] (\l, {ln((2*(1-2/\l))/(sqrt(1 + (4*(1-2/\l)*(1-2/\l))) - 1))});
                  \addplot[domain=3:50, dotted, variable=\l] (\l, {ln(2/(sqrt(5)-1))});
            \end{axis}
        \end{tikzpicture}
        \caption{}
        \label{fig:achievable}
    \end{subfigure}
    \vspace{-1cm}
    \caption{(a) Leakage privacy as a function of $\geoparam$ for $t=2$.
    (b) Lower bound on leakage privacy.
    }
\end{figure}
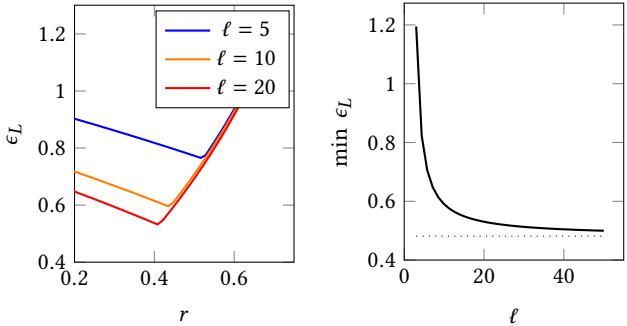


%% file: Related.tex
\paragraph{Differentially Private Statistics over Key-Value Data}
Our work is most closely related to prior literature computing key-value analytics under the LDP model. The first work in this space was PrivKV from Ye et al.~\cite{ye19}. Ye et al.~modify the Harmony randomized response based protocol~\cite{nguyn2016collecting} to better maintain the relationships between the keys and values. They apply their baseline technique in an iterative manner, called PrivKVM, where each iteration gets closer to the true mean. To reduce communication costs, PrivKVM requires each client only send a single randomly sampled key-value pair. This sampling comes at the cost of additional sampling error in the final result. Follow-up work from Gu et al.~\cite{gu2019pckv} introduced PCKV, which improves upon PrivKVM in several aspects. They apply an advanced sampling procedure to enhance utility over the sampling done by PrivKVM. They also require only a single iteration and provide a tighter analysis of the privacy budget consumption. Sun et al.~investigate alternate perturbation techniques under the PrivKV framework and investigate conditional analysis of key-value data~\cite{sun2019conditional}.

While this line of work has made substantial advancements to the computation of key-value statistics in the local trust model, it suffers an inherent limitation. In any LDP solution, the randomness is introduced to each value before aggregation and thus the error compounds in the final result. As we show in our evaluation, this increases the error by a factor of $O(\sqrt{m})$. Bittau et al.~\cite{frist_shuffle} investigate reducing the effects of compounding noise by introducing the shuffle model of DP. The shuffle model lies in between the local and central models, with a weaker trust assumption than the central model but improved utility guarantees over the local model. In this model, a semi-trusted shuffler lies between the participants and the untrusted curator. The shuffler introduces anonymity by performing a permutation on the noisy data it receives from the participants before handing it to the curator. Since the initial work from Bittau et al.~\cite{frist_shuffle}, there have been many interesting works in this space~\cite{privacy_blanket, google_shuffle,ullman_shuffle}. Although these works significantly reduce the dependency on the number of clients (for example Balle et al.'s work has a dependency of $O(\sqrt[6]{m})$ \cite{privacy_blanket}), they still involve adding noise locally, which in addition to making the error dependent on the number of clients, disrupts the correlation between keys and values. It has been shown that the utility of this approach is strictly between that of the local and central models~\cite{ullman_shuffle,privacy_blanket}. Thus in our work, we take the approach of simulating the central trust model using MPC to avoid any dependence on the number of clients by adding randomness in a central manner. 

\paragraph{Combining Cryptography and Differential Privacy}
The area of research dedicated to the intersection of cryptography and differential privacy has been termed DP-Cryptography \cite{dpcryptography}. One direction of this research uses MPC when computing differentially private statistics to eliminate the need to trust a central data curator \cite{eigner14,dwork06,djoin}. Dwork et al.~\cite{dwork06} provided protocols to generate Gaussian and exponentially distributed noise in a distributed setting. The combination of secure computation and DP was further explored in the construction of PrivaDA~\cite{eigner14}, a generic design that allowed for noise generation in a fully distributed setting. We note that these constructions cannot immediately be applied to key-value data without incurring high computation costs as a result of secret sharing the keys. Pettai and Laud~\cite{pettai15} implemented the DP sample-and-aggregate method in an MPC framework, showing that the combination can still provide reasonable performance. Later work used homomorphic encryption for the secure computation protocols~\cite{acar17,li17,shi11,roy2020crypt}. Goryczka et al.~\cite{goryczka13} implemented and compared the security, performance, and scalability of early models that combined MPC and DP.

Another area of DP-Cryptography uses DP to improve the performance of cryptographic primitives. Allowing for differentially private leakage can improve the computation or communication costs. Chen et al.~\cite{chen18} employ this technique in the context of searchable symmetric encryption. Toledo et al.~\cite{toledo16} relax the information-theoretic requirements of private information retrieval (PIR) to provide more efficient constructions of PIR which satisfy differentially private leakage. Similarly, Shrinkwrap~\cite{shrinkwrap} uses differentially private query processing to improve the performance of evaluating SQL queries. Mazloom and Gordon~\cite{mazloom18} extend the concept of differentially private leakage to secure multi-party computations of histograms, PageRank, and matrix factorization. DP leakage has also been used in the context of anonymous communication to reduce communication costs \cite{vuvuzela,stadium,karaoke}. Our work is unique in that we simultaneously take both approaches to combining DP and cryptography: we allow differentially private leakage to improve performance and use secure computation to eliminate trust assumptions.

%% file: Conclusion.tex
Although local differentially private solutions are becoming increasingly popular, they are also well known to provide poor utility. Unfortunately, a common way to circumvent this issue is to let $\epsilon$ be unreasonably large, offering almost meaningless protection to users' privacy. By minimizing the costs that organizations incur when choosing privacy-aware solutions, we can encourage the development and use of systems that respect users' privacy.

In line with this ideal, we provide an improved solution for calculating the frequency and mean of key-value data in the local trust model. Our work presents a novel approach to MPC, which circumvents traditional efficiency issues. Simultaneously, our approach enables us to achieve pure differential privacy guarantees on leakage without removing any data. By leveraging cryptographic primitives, we maintain the correlation between keys and values, resulting in a high accuracy algorithm. We believe that offering high efficiency, accuracy, privacy, and conceptual simplicity incentivizes the adoption of our approach in practice.

%% file: appendix.tex
\section{Proof of Privacy}\label{sec:privacy}
\privacythm*

\begin{proof}
We begin by expanding the statement needed to prove differential privacy. First, we derive
\begin{equation}\label{eqn:reduction}\small
    \frac{\Pr[\leakage(\allkv)=\observation]}{\Pr[\leakage(\allkv')=\observation]}=\frac{\prod\limits_{k=0}^\numkeys \Pr[\leakage_k(\allkv)=\observation_k]}{\prod\limits_{k=0}^\numkeys \Pr[\leakage_k(\allkv')=\observation_k]}\leq \left(\frac{\Pr[\leakage_j(\allkv)=\observation_j]}{\Pr[\leakage_j(\allkv')=\observation_j]}\right)^\distinctkeys
\end{equation}
where $j$ is one of at most $\distinctkeys$ keys that differ between $\allkv$ and $\allkv'$.

We now consider the two worst cases of neighbouring inputs for a specific key $j$ under our assumption that $|\freq_j(\allkv)-\freq_j(\allkv')|\leq 1$. These cases are when $\freq_j(\allkv) - 1 = \freq_j(\allkv')$ and when $\freq_j(\allkv)+1=\freq_j(\allkv')$. First, we consider the case that $\freq_j(\allkv) - 1 = \freq_j(\allkv')$.
In this case we have that
\begin{align}
    \frac{\Pr[\leakage_j(\allkv)=\observation_j]}{\Pr[\leakage_j(\allkv')=\observation_j]}&=
   \frac{\sum\limits_{v=0}^\infty \mathbb{G}(v;\geoparam) \mathbb{B}(\observation_j;\freq_j(\allkv)+v,\subsetratio)}
   {\sum\limits_{v=0}^\infty \mathbb{G}(v;\geoparam) \mathbb{B}(\observation_j;\freq_j(\allkv')+v,\subsetratio)}\\
      =& \frac{\sum\limits_{v=1}^\infty \mathbb{G}(v-1;\geoparam) \mathbb{B}(\observation_j;\freq_j(\allkv')+v,\subsetratio)}
   {\sum\limits_{v=0}^\infty \mathbb{G}(v;\geoparam) \mathbb{B}(\observation_j;\freq_j(\allkv')+v,\subsetratio)}\label{eqn:assump_1},
\end{align}
where (\ref{eqn:assump_1}) follows from a change of variables $v=v-1$ and our assumption that $\freq_j(\allkv) - 1 = \freq_j(\allkv')$.
Next we apply the fact that $\mathbb{G}(v-1;\geoparam)=\frac{1}{1-\geoparam}\mathbb{G}(v;\geoparam)$, for $v\geq1$, to obtain
\begin{align}
    \frac{\Pr[\leakage_j(\allkv)=\observation_j]}{\Pr[\leakage_j(\allkv')=\observation_j]}&= \frac{\sum\limits_{v=1}^\infty \frac{1}{1-\geoparam}\mathbb{G}(v;\geoparam) \mathbb{B}(\observation_j;\freq_j(\allkv')+v,\subsetratio)}
   {\sum\limits_{v=0}^\infty \mathbb{G}(v;\geoparam) \mathbb{B}(\observation_j;\freq_j(\allkv')+v,\subsetratio)}\\
   &\leq \frac{1}{1-\geoparam}\label{eqn:case_1},
\end{align}
where the last inequality follows by widening the summation limit to include zero and cancelling like terms. We note that this gives us the first term in the maximum from (\ref{eqn:epsilon_statement}).

Next we consider the case where $\freq_j(\allkv)+1=\freq_j(\allkv')$. The first steps proceed similarly except we perform the change of variables $v=v+1$. That is,
\begin{align}\hspace*{-1em}
    \frac{\Pr[\leakage_j(\allkv)=\observation_j]}{\Pr[\leakage_j(\allkv')=\observation_j]}&=
   \frac{\sum\limits_{v=0}^\infty \mathbb{G}(v;\geoparam) \mathbb{B}(\observation_j;\freq_j(\allkv)+v,\subsetratio)}
   {\sum\limits_{v=0}^\infty \mathbb{G}(v;\geoparam) \mathbb{B}(\observation_j;\freq_j(\allkv')+v,\subsetratio)}\\
   =&\ \frac{\sum\limits_{v=-1}^\infty \mathbb{G}(v+1;\geoparam) \mathbb{B}(\observation_j;\freq_j(\allkv')+v,\subsetratio)}
   {\sum\limits_{v=0}^\infty \mathbb{G}(v;\geoparam) \mathbb{B}(\observation_j;\freq_j(\allkv')+v,\subsetratio)}\\
    =&\ \frac{\mathbb{G}(0;\geoparam)\mathbb{B}(\observation_j;\freq_j(\allkv')-1,\subsetratio)}
   {\sum\limits_{v=0}^\infty \mathbb{G}(v;\geoparam) \mathbb{B}(\observation_j;\freq_j(\allkv')+v,\subsetratio)} + (1-\geoparam)\label{eqn:breakout},
\end{align}
where (\ref{eqn:breakout}) follows from separating off the first term of the sum in the numerator and applying the fact that $\mathbb{G}(v+1;\geoparam)=(1-\geoparam)\mathbb{G}(v;\geoparam)$, for $v\geq0$, to cancel like terms. 
Next, we observe that when $\observation_j\geq \freq_j(\allkv')$, (\ref{eqn:breakout}) simplifies to $(1-\geoparam)$ since the binomial coefficient in the numerator is zero. Hence, all that remains is the case where $\observation_j<\freq_j(\allkv')-1$. In this case, we get the following inequality by decreasing the upper limit of the sum in the denominator to zero. Specifically,
\begin{align}
     \frac{\Pr[\leakage_j(\allkv)=\observation_j]}{\Pr[\leakage_j(\allkv')=\observation_j]}&\leq \frac{\mathbb{B}(\observation_j;\freq_j(\allkv')-1,\subsetratio)}
   { \mathbb{B}(\observation_j;\freq_j(\allkv'),\subsetratio)} + (1-\geoparam)\label{eqn:loose_part}\\
       =&\ \left(\frac{1}{1-\subsetratio}\right)\frac{\binom{\freq_j(\allkv')-1}{\observation_j}}
   {\binom{\freq_j(\allkv')}{\observation_j}} +
   (1-\geoparam)\label{eqn:expand_binom}\\
   =&\ \left(\frac{1}{1-\subsetratio}\right)\frac{\freq_j(\allkv')-\observation_j}
   {\freq_j(\allkv')} +
   (1-\geoparam)\label{eqn:cancel_binom}\\
      \leq&\ \frac{1}{1-\subsetratio} + 1-\geoparam\label{eqn:case_2},
\end{align}
where (\ref{eqn:expand_binom}) comes from applying (\ref{eqn:pdf_bino}) and (\ref{eqn:cancel_binom}) applies the definition of the binomial coefficient. Finally, (\ref{eqn:case_2}) follows from the fact that $0\leq \observation_j<\freq_j(\allkv')-1$.

Combining the two cases from (\ref{eqn:case_1}) and (\ref{eqn:case_2}) gives
\begin{equation}
    \frac{\Pr[\leakage_j(\allkv)=\observation_j]}{\Pr[\leakage_j(\allkv')=\observation_j]} \leq \max\left\{ \frac{1}{1-\geoparam},\frac{1}{1-\subsetratio}+1-\geoparam \right\}.
\end{equation}
Substituting this result into (\ref{eqn:reduction}) from above, it follows that
\begin{align}
    &\frac{\Pr[\leakage(\allkv)=\observation]}{\Pr[\leakage(\allkv')=\observation]} \leq \left(\frac{\Pr[\leakage_j(\allkv)=\observation_j]}{\Pr[\leakage_j(\allkv')=\observation_j]}\right)^\distinctkeys \\
    \leq&\ 
    \left(\max\left\{ \frac{1}{1-\geoparam},\frac{1}{1-\subsetratio}+1-\geoparam \right\}\right)^\distinctkeys = e^{\leaknoise},
\end{align}
where
 \begin{equation}\label{eqn:leak_epsilon}
     \leaknoise = \distinctkeys\ln{\left( \max\left\{ \frac{1}{1-\geoparam},\frac{1}{1-\subsetratio}+1-\geoparam \right\}\right)}.
 \end{equation}
\end{proof}

\section{Privacy of the Output}\label{sec:output_privacy}
\paragraph{Frequency Estimation.}
We denote the output of the frequency estimation protocol by $\freqalg(\allkv)$. The output of this algorithm is $\freq(\allkv) + Lap(\Delta/\freqnoise)^\numkeys$, where $\Delta=\distinctkeys$ and $Lap(\Delta/\freqnoise)^\numkeys$ represents $\numkeys$ i.i.d draws from the Laplace distribution with parameter $\Delta/\freqnoise$.
\begin{theorem}
    The algorithm $\freqalg$ satisfies $\freqnoise$-DP.
\end{theorem}
\begin{proof}
    Each client has at most one copy of each key, so we have that $|\freqk(\allkv) - \freqk(\allkv')|\leq 1$. Additionally, we know that each client can have at most $\distinctkeys$ distinct keys. Hence,
    \begin{equation}
        \Delta=\max_{\allkv,\allkv'}\|\freq(\allkv)-\freq(\allkv')\|_1 = \distinctkeys.
    \end{equation}
    The result then follows trivially from the fact that the Laplace mechanism is differentially private.
\end{proof}
\paragraph{Mean Estimation.}
Recall that the output of the mean estimation protocol is denoted by $\meanalg(\allkv)$. Denote the true mean of the dataset $\allkv$ to be $\mean(\allkv)$. Thus, the output of the algorithm $\meanalg$ is $\mean(\allkv) + Lap(\Delta/\meannoise)^\numkeys$ where $\Delta=\frac{\distinctkeys 2R}{\minfreq}$ and $Lap(\Delta/\freqnoise)^\numkeys$ represents $\numkeys$ i.i.d draws from the Laplace distribution.
\begin{theorem}
    The algorithm $\meanalg$ satisfies $\meannoise$-DP.
\end{theorem}
\begin{proof}
    Each client has at most one value for each key and each value is in $[x-R,x+R]$, so it follows that $|\meank(\allkv) - \meank(\allkv')|\leq \frac{2R}{\freqk(\allkv)}$. Recall that $\minfreq=\min_{k\in[\numkeys]}{\freqk(\allkv)}$. Thus, we have that $\forall k\in[\numkeys]$, $|\meank(\allkv) - \meank(\allkv')|\leq \frac{2R}{\minfreq}$. Additionally, we know that each client has at most $\distinctkeys$ distinct keys. Therefore,
    \begin{equation}
        \Delta=\max_{\allkv,\allkv'}\|\mean(\allkv)-\mean(\allkv')\|_1 \leq \frac{\lambda 2R}{\minfreq}.
    \end{equation}
    Then privacy follows trivially from the fact that the Laplace mechanism is differentially private.
\end{proof}

\section{Proof of Security}
\label{sec:proof-security}

\securitythm*

\begin{proof}
Consider the view of a probabilistic polynomial-time (in $\lambda$) adversary $\mathcal{A}$ with input $D_A$, $\textsc{VIEW}^{\Pi}_A(D_A,\cdot)$. Here, we let $A$ correspond to a set of $b_{N} = t-1$ computation nodes wherein the multi-party computation protocol uses a $(t,t)$-threshold scheme. Then, $D_A$ denotes the data that each computation node in the set $A$ receives from the honest clients, which we can represent with a histogram of keys. Note that the views of the clients that the adversary controls is trivially simulated since clients are only responsible for submitting their own data, so the corresponding view is empty. Let $(D,D')$ be neighbours with respect to $f(D_A,\cdot)$. Let $\Pi_1$ denote the protocol consisting of the clients and dummy generator sending data to the computation nodes and $\Pi_2$ denote the multi-party computation protocol to compute the function $f$. Then $\Pi$ consists of applying these two protocols sequentially. The process of sending data between clients and the computation nodes does not leak any additional information since we assume anonymous, secure channels. The computation nodes only see the shares they are meant to receive and cannot determine their origin. Further, by the information-theoretic security of the threshold scheme, $t-1$ computation nodes do not learn any additional information by compiling their shares together. By our assumption, any metadata associated with the data is stripped. By \Cref{thm.priv}, it immediately follows that $\Pi_1$ satisfies IND-CDP-MPC. That is,
\begin{align*}
&\Pr[\mathcal{A}(\textsc{view}_A^{\Pi_1}(D_A, D))=1] \\
\leq &\exp(\epsilon_L) \cdot \Pr[\mathcal{A}(\textsc{view}_A^{\Pi_1}(D_A, D'))=1].
\end{align*}
Also, given that the multi-party computation protocols ensure statistical privacy, we can simulate the process with an oracle whose output is statistically indistinguishable from the real values. That is, each computation node gains some information from the input and subsequent steps only reveal information indistinguishable from random. This reveals some negligible information in the security parameter $\lambda$. Since the function $f$ satisfies $\epsilon_f$-DP, it follows that $\Pi_2$ satisfies IND-CDP-MPC, i.e.,
\begin{align*}
&\Pr[\mathcal{A}(\textsc{view}_A^{\Pi_2}(D_A, D))=1] \\
\leq &\exp(\epsilon_f) \cdot \Pr[\mathcal{A}(\textsc{view}_A^{\Pi_2}(D_A, D'))=1] + negl(\lambda).
\end{align*}
Thus, given that both $\Pi_1$ and $\Pi_2$ ensure IND-CDP-MPC and $\Pi_2$ can be treated as a black-box, we can now consider the composition of the two protocols. The probability of distinguishing $D$ and $D'$ is bounded by
\begin{align*}
&\Pr[\mathcal{A}(\textsc{view}_A^{\Pi_2,\Pi_1}(D_A, D))=1] \\
= &\int_x (\Pr[\mathcal{A}(\textsc{view}_A^{\Pi_2}(D_A, D,x))=1]) \cdot \Pr[x = \textsc{view}_A^{\Pi_1}(D_A, D)]\,dx \\
\leq &\int_x (\exp(\epsilon_f)\Pr[\mathcal{A}(\textsc{view}_A^{\Pi_2}(D_A,D',x))=1] + negl(\lambda)) \\
&\cdot (\exp(\epsilon_L)\Pr[x = \textsc{view}_A^{\Pi_1}(D_A, D')])\,dx \\
= &\exp(\epsilon_L+\epsilon_f)\Pr[\mathcal{A}(\textsc{view}_A^{\Pi_2,\Pi_1}(D_A, D'))=1] + negl(\lambda).
\end{align*}
\end{proof}

\section{Privacy of the One-Sided Dummy Solution}\label{app:one_sided_dummy}
For a description of the one-sided dummy solution, see \Cref{sec.early}. We consider an implementation of our one-sided dummy solution using the strategy by Mazloom et al.~\cite{mazloom18}. This strategy uses a modified version of the Two-sided Geometric distribution (an over-approximation) to add noise. For completeness, we provide a version of their proofs, in our context, using the standard two-sided Geometric distribution.

\subsection{Background on the Geometric Mechanism}
First, we recall that the PMF of the two-sided geometric distribution is
\begin{equation}
    \mathbb{G}_2(z;\alpha) = \frac{1-\alpha}{1+\alpha} \alpha^{|z|}
\end{equation}
for $\alpha \in [0,1]$.

\begin{definition}[Geometric Mechanism]
    Let $f:\mathcal{X}^n\mapsto \mathbb{R}^k$. The Geometric Mechanism is defined as 
    \begin{equation}
        G(X) = f(X) + \mathbb{G}_2(e^{-\epsilon/\Delta})^\numkeys,
    \end{equation}
    where $\mathbb{G}_2(e^{-\epsilon/\Delta})^\numkeys$ represents $\numkeys$ i.i.d draws from the Geometric distribution with parameter $\alpha = e^{-\epsilon/\Delta}$
\end{definition}
\begin{theorem}[\cite{ghosh2012universally}]\label{DP of geometric}
    The Geometric Mechanism satisfies $\epsilon$-DP
\end{theorem}
\noindent\textbf{Proof Sketch}
\begin{equation}
    \frac{\mathbb{G}_2(z;\alpha)}{\mathbb{G}_2(z+\Delta;\alpha)} =
    \frac{\frac{1-\alpha}{1+\alpha} \alpha^{|z|}}{\frac{1-\alpha}{1+\alpha} \alpha^{|z+\Delta|}} \leq 
    \alpha^{-|\Delta|} = e^\epsilon
\end{equation}
where the inequality comes from triangle inequality and the last equality from setting $\alpha = e^{-\epsilon/\Delta}$.

\subsection{Privacy proof}
We consider the one-sided dummy mechanism, defined as 
\begin{equation}
    D_k(\allkv) = q_k(\allkv) + max(0,\beta + \mathbb{G}_2(e^{-\epsilon}))
\end{equation}
where $\beta$ is a constant chosen such that $Pr[\beta + \mathbb{G}_2(e^{-\epsilon})<0]$ is negligible.
\begin{theorem}[\cite{mazloom18}]
     $D_k(\allkv)$ is $(\epsilon, \delta)$-DP
\end{theorem}
\begin{proof}
Using Theorem~\ref{DP of geometric} we can obtain the approximate DP guarantee as follows. First, we define a set $\mathcal{B}$ of all outcomes where we add negative noise. Conditioning on not being in this set we get
\begin{eqnarray}
    Pr[D(\allkv) \in Z\setminus \mathcal{B}] &=& \prod_{k=0}^\numkeys Pr[D_k(\allkv) \in Z_k\setminus \mathcal{B}]\\
    &\leq& \prod_{k=0}^\numkeys e^\epsilon Pr[D_k(\allkv') \in Z_k\setminus \mathcal{B}]\\
    &=& e^\epsilon Pr[D(\allkv') \in Z\setminus \mathcal{B}]
\end{eqnarray}

Now, using this, we obtain the final result
\begin{align}
    Pr[D(\allkv) \in Z] &= Pr[D(\allkv) \in Z\setminus \mathcal{B}] + Pr[D(\allkv) \in \mathcal{B}]\\
    &\leq e^\epsilon Pr[D(\allkv') \in Z\setminus \mathcal{B}] + Pr[D(\allkv) \in \mathcal{B}]\\
    &\leq e^\epsilon Pr[D(\allkv') \in Z] + Pr[D(\allkv) \in \mathcal{B}]
\end{align}

Thus as long as we choose $\beta$ such that $Pr[D(\allkv) \in \mathcal{B}] \leq \delta$, we have approximate DP.
\end{proof}
\subsection{How many dummies to add}
\label{sec:dummies-to-add}
Next, we specify how to choose $\beta$ to satisfy approximate DP. To do this, we begin by calculating
\begin{align}
    Pr[\beta + \mathbb{G}_2(\alpha)<0]&= \sum\limits_{z=-\infty}^0 \frac{1-\alpha}{1+\alpha} \alpha^{|z-\beta|}\\
    &=\sum\limits_{z=2\beta}^\infty \frac{1-\alpha}{1+\alpha} \alpha^{|z-\beta|}\\
    &=\frac{1-\alpha}{1+\alpha} \alpha^\beta \sum\limits_{z=0}^\infty  \alpha^z\\
    &=\frac{\alpha^\beta}{1+\alpha}
\end{align}
then, applying the union bound we get that
\begin{equation}\label{eqn:delta_case}
    Pr[D(\allkv) \in \mathcal{B}] \leq \frac{\lambda \alpha^\beta}{1+\alpha} \leq \delta.
\end{equation}
Finally, we can rearrange (\ref{eqn:delta_case}) to obtain a formula for $\beta$
\begin{equation}
    \beta > \log_\alpha\left(\frac{\delta(1+\alpha)}{\lambda}\right)=-\frac{1}{\epsilon}\ln\left(\frac{\delta(1+e^{-\epsilon})}{\lambda}\right).
\end{equation}

\subsection{Comparison with our protocol}\label{app:compare}
\begin{figure}[ht]
\centering
		\includegraphics[scale=0.5]{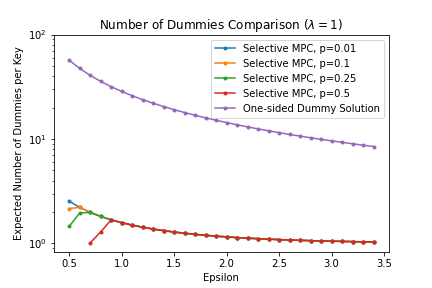}
	    \caption{We consider the expected number of dummies for various epsilons. $r$ is chosen to minimize the epsilon spent
	    and where applicable delta is fixed to $\delta=2^{-40}$.} 
		\label{fig:expt_dummies}
\end{figure}
The primary difference between the one-sided dummy solution and our final protocol is how the number of dummies is selected. In the one-sided dummy solution, dummies are chosen from a geometric distribution shifted by some amount, which we call $\beta$. The expected value of a single sample of the shifted and truncated two-sized geometric distribution is at least $\beta$. Additionally, each element is given to all nodes rather than a subset of them. Considering these changes, the client-to-server communication of the one-sided dummy solution is $\Omega\left(\numkv\compnodes+\numkeys\beta\compnodes\right)$. The computation cost of this protocol in the MPC phase is the same as our solution.

Figure \ref{fig:expt_dummies} demonstrates the difference in the expected number of dummies for the one-sided dummy solution and our final solution (selective MPC). It is clear that selective MPC produces less dummies and is more efficient.

\change{
\subsection{Using the noise from dummies}
A natural question is if the dummies used to achieve a DP leakage could also be used to make the output private (instead of adding additional Laplace noise in a distributed manner). That is, instead of discarding the dummies, we may count them as real values to add noise to the frequency statistic. We remark that this is not possible in selective MPC. Due to the fact that we add dummies following a one-sided geometric distribution, the output cannot satisfy DP (even approximate DP as the $\delta$ is prohibitively large). However, the dummies could be used in the one-sided dummy solution. Specifically, for frequency estimation, we have already paid the privacy cost in $\leaknoise$, and could simply publish the view of the computation node as the noisy frequency. One would likely want to subtract $\beta$ from the final result to improve accuracy. This is private since $\beta$ is derived from public values $\epsilon$, $\delta$, and $\lambda$. The downside to this approach is that it inherits the downsides of the one-sided dummy solution: approximate DP instead of pure DP and increased communication due to a large number of dummies. It is also not clear how this approach could be extended to mean estimations.
}

\section{Additional Experiments}\label{app.input_val_exp}
\paragraph{Input Validation}
\change{We examine the cost of validating the client's input. We implement the $x\cdot(1-x) \stackrel{?}{=}0$ check described in \Cref{sec.malicious}, but this can easily be extended to checking arbitrary constraints. \Cref{tab:input_valid_time} shows the time to perform input validation using 3 computation nodes. The input validation can be performed in parallel so we conduct the experiment for 10, 100, 1000, and 10,000 input data points using the Amazon EC2 C5 instances described in \Cref{sec:experiments}. Similar to the other experiments, we use local, remote, and distant network setups. The runtimes for input validation remain reasonable, particularly compared to the runtime of the computation phase, which suggests that input validation does not impose an impractical burden on the protocol.}
\begin{figure}
\begin{tabular}{cccc}
        \toprule
        \# Inputs & Local & Remote & Distant \\
        \midrule
        10 & 0.08~(0.00) & 1.01~(0.07) & 2.02~(0.13) \\ \midrule
        100 & 0.80~(0.03) & 1.31~(0.06) & 2.38~(0.12) \\ \midrule
        1000 & 0.73~(0.02) & 2.02~(0.09) & 3.75~(0.16) \\ \midrule
        10000 & 0.86~(0.05) & 9.47~(0.15) & 18.34~(0.14) \\
        \bottomrule
    \end{tabular}
    \captionof{table}{Run time (in seconds) to validate multiple inputs with 3 computation nodes using maliciously secure MPC. Measurements are means with standard deviations in parentheses.}
    \label{tab:input_valid_time}
\end{figure}